\documentclass[sigconf]{acmart}

\setcopyright{none} 
\renewcommand\footnotetextcopyrightpermission[1]{} 

\usepackage{color}	
\usepackage{xspace}
\usepackage[table]{xcolor}
\usepackage{soul}

\usepackage{graphicx}
\usepackage{color}
\usepackage{enumerate}
\usepackage{hyperref}
\usepackage[font={small,sf}]{caption}
\usepackage{floatflt}
\usepackage{lipsum}
\usepackage{multirow}
\usepackage{soul}
\sethlcolor{black!10}
\usepackage{mdframed}
\usepackage{amsmath,amsfonts}
\usepackage{comment}
\usepackage{subcaption}
\usepackage{enumitem}
\usepackage{csquotes}
\usepackage{tabularx}
\usepackage{tabularx}
\usepackage{enumitem}
\usepackage{tikz}
\usetikzlibrary{positioning}
\usepackage{mathtools}
\usepackage[a-1b]{pdfx}  
\usepackage[T1]{fontenc}
\usepackage{microtype}
\usepackage{bbm}
\usepackage{fnpct}  
\usepackage{tcolorbox}

\graphicspath{{./figures/}}

\usepackage{algorithm}
\usepackage[noend]{algpseudocode}

\DeclarePairedDelimiterX\set[1]\lbrace\rbrace{\,#1\,}

\usepackage{tikz}
\usetikzlibrary{positioning}

\newcommand*{\Reals}{\mathbb{R}}

\newcommand*{\Naturals}{\mathbb{N}}

\DeclareMathOperator*{\argmin}{arg\,min}

\newcommand{\stitle}[1]{\vspace{1mm}\noindent{\textbf{#1}}.}

\newcommand{\dee}{\mathcal{D}}

\newcommand{\eps}{\varepsilon}

\newcommand{\dual}{\mathsf{dual}}

\newtheorem{theorem}{Theorem} 
\newtheorem{lemma}[theorem]{Lemma}

\newtheorem{problem}{Problem}
\newtheorem{example}{Example}

\newcommand{\data}{\mathcal{D}\xspace}
\newcommand{\ef}{f}
\newcommand{\score}{\mathtt{score}\xspace}
\newcommand{\point}{p}
\newcommand{\rank}{\mathtt{rank}\xspace}
\newcommand{\nnrank}{\mathtt{nn\_rank}\xspace}
\newcommand{\ith}{\textsc{RAR}\xspace} 

\newcommand{\srs}{\textsc{SRR}\xspace} 
\newcommand{\Stripe}{\mathcal{S}}
\newcommand{\conf}{\mathcal{C}}
\newcommand{\result}{\mathcal{D}_o}

\newcommand{\epssample}{\mathcal{N}}
\newcommand{\Ball}{\mathcal{B}}
\newcommand{\ball}{\mathcal{B}}
\newcommand{\similarity}{\textsc{RAS}}
\newcommand{\css}{\textsc{CSS}}
\newcommand{\dist}{\mathbf{d}}

\newtcolorbox{highlightbox}{
  colback=blue!10,
  colframe=blue!20,
  arc=0.5mm, 
  fonttitle=\bfseries,
  boxrule=0mm,
  boxsep=0mm,
  left=0mm,
  right=0mm,
  top=0mm,
  bottom=0mm
}

\newtcolorbox{examplebox}{
  colback=blue!10,
  colframe=blue!20,
  arc=2mm, 
  fonttitle=\bfseries,
  boxrule=0mm,
  boxsep=1mm,
  left=0mm,
  right=0mm,
  top=0mm,
  bottom=0mm
}

\newtcolorbox{defbox}{
  colback=orange!10,
  colframe=orange!20,
  arc=2mm, 
  fonttitle=\bfseries,
  boxrule=0mm,
  boxsep=1mm,
  left=0mm,
  right=0mm,
  top=0mm,
  bottom=0mm
}

\newtcolorbox{pbox}{
  colback=black!5,
  colframe=black!30,
  arc=2mm, 
  fonttitle=\bfseries,
  boxrule=0mm,
  boxsep=1mm,
  left=0mm,
  right=0mm,
  top=0mm,
  bottom=0mm
}

\begin{document}
\title{Random-Access Ranked Retrieval and Similarity Search} 

\author{Mohsen Dehghankar}
\affiliation{%
  \institution{University of Illinois Chicago}
  \city{Chicago}
  \state{Illinois}
  \country{USA}
}
\orcid{0009-0006-1687-8012}
\email{mdehgh2@uic.edu}

\author{Abolfazl Asudeh}
\affiliation{%
  \institution{University of Illinois Chicago}
  \city{Chicago}
  \state{Illinois}
  \country{USA}
}
\orcid{0000-0002-5251-6186}
\email{asudeh@uic.edu}

\author{Raghav Mittal}
\affiliation{%
  \institution{University of Texas at Arlington}
  \city{Arlington}
  \state{Texas}
  \country{USA}
}
\orcid{0000-0002-4705-0123}
\email{rxm0006@mavs.uta.edu}

\author{Suraj Shetiya}
\affiliation{%
  \institution{IIT Bombay}
  \city{Mumbai}
  \state{Maharastra}
  \country{India}
}
\orcid{0000-0001-9166-2365}
\email{surajs@cse.iitb.ac.in}

\author{Gautam Das}
\affiliation{%
  \institution{University of Texas at Arlington}
  \city{Arlington}
  \state{Texas}
  \country{USA}
}
\orcid{0000-0002-4627-9065}
\email{gdas@uta.edu}

\begin{abstract}
We extend {\em Random Access}, a fundamental operation that enables efficient search and exploration algorithms, to the modern interactive data systems based on {\em Ranked Retrieval} and {\em Similarity Search}, where orderings are dynamically defined over a {\em high-dimensional} feature space. 
This extension enables efficient solutions for a wide range of applications, from data analytics tools and database systems to recommendation systems and machine learning.

We formalize the \emph{Random-Access Ranked Retrieval} (\textsc{RAR}) problem, and extend it to Similarity Search.
Our algorithmic innovations include the development of a theoretically efficient algorithm based on geometric arrangements, achieving logarithmic query time. However, this method suffers from exponential space complexity in high dimensions.
Therefore, we develop a second class of algorithms based on $\varepsilon$-sampling, which consume a linear space. Since exactly locating the tuple at a specific rank is challenging due to its connection to the range counting problem, we introduce a relaxed variant called \emph{$\kappa$-Random-Access Ranked Retrieval}, which returns a small subset of size $\kappa$ guaranteed to contain the target tuple.
To solve this problem efficiently, we define an intermediate problem, \emph{Stripe Range Retrieval} (\textsc{SRR}), and design a hierarchical sampling data structure tailored for narrow stripe range queries. Our method achieves practical scalability in both data size and dimensionality.
We prove near-optimal bounds on the efficiency of our algorithms and validate their performance through extensive experiments on real and synthetic datasets, demonstrating scalability to millions of tuples and hundreds of dimensions.
\end{abstract}

\maketitle




\section{Introduction}
Random access to data is a foundational operation in computer science, enabling efficient algorithms for search, selection, and exploration. However, this operation largely breaks down in modern data systems where access is governed by ranking functions rather than static keys. In such systems, data items are ordered by a {\em query-specific}, {\em user-specific}, or {\em machine-generated} scoring function\footnote{Also known as ranking functions.}, often {\em dynamically} defined over {\em high-dimensional} feature spaces. Yet, accessing the items at {\em arbitrary} rank positions or within a rank window typically requires fully materializing or sequentially scanning the ranked list, using operations such as {\tt get-next}.

This limitation arises across a wide range of data-driven applications, including database systems, interactive data exploration and analytics tools, information retrieval and recommendation systems, and machine learning pipelines~\cite{eldar2023direct, baeza1999modern, croft2010search, brin1998anatomy, broder2002taxonomy, ricci2021recommender}. Common scenarios include database queries on derived attributes\cite{ma2024spreadsheetbench,levandoski2010preference}, analytical exploration for business decisions~\cite{ma2024spreadsheetbench}, paged navigation of search results~\cite{croft2010search}, browsing ranked product lists in e-commerce~\cite{orlivskyi2021pagination}, and scrolling through personalized feeds on social platforms~\cite{li2010personalized}. 

A concrete example arises in web application APIs, which implement this through paged or offset-based navigation:

\begin{example}\label{eg-1}
    In web and e-commerce applications, APIs commonly use {\em offset-based pagination}. 
    For instance, a request such as {\tt [GET /items?offset=400\&limit=5]} retrieves items ranked 401--405. 
    More generally, this pattern requires efficient retrieval of results 
    at deeper ranks without scanning all earlier entries~\cite{orlivskyi2021pagination, sudda2024optimizing}.
\end{example}

Similar scenarios also arise in data science and visual analytics tooling, where analysts frequently interact with large, ranked datasets. Tools such as spreadsheets, notebooks, and visual analytics platforms provide intuitive interfaces for exploration~\cite{chen2023visualizing,rahman2020benchmarking,shankar2022bolt,zhao2023data}, but the explosive growth of data volumes driven by big data and AI has placed heavy demands on their efficiency and responsiveness~\cite{bikakis2025visual}.
For example, while using data science workflows such as Jupyter notebooks~\cite{kluyver2016jupyter}, analysts may need random access to particular ranks or quantiles (e.g., 10\textsuperscript{th}, 20\textsuperscript{th}, etc.) in a sorted dataframe, effectively \emph{skipping} earlier-ranked items.

Beyond ranked retrieval, closely related challenges arise in similarity search, where items are ordered by their distance or similarity to a query object.

\begin{example}\label{eg-3}
In similarity search applications, such as nearest-neighbor queries over high-dimensional embeddings, items are ranked by their similarity to a query point.
While existing systems efficiently retrieve the top-$k$ most similar items, users and downstream applications may need to access items at arbitrary similarity ranks or within specific distance bands.
For example, an analyst may wish to examine the points whose similarity lies within a narrow range around a decision boundary.
\end{example}


Traditionally, research on ranked retrieval has focused on \emph{top-$k$} query processing, where the goal is to identify the highest-ranked tuples~\cite{ilyas2008survey, fagin, gautopk, ilyasrank}. In this setting, results are retrieved sequentially from the top of the ranking up to position $k$, where $k$ is typically a small constant relative to the dataset size. 
Similarly, in the context of similarity search, $k$-nearest neighbor (kNN) search aims to retrieve the $k$ most similar points to a query~\cite{wang2021comprehensive, pan2024survey, malkov2018efficient}. These approaches are designed for small values of $k$ and, as shown in our evaluations, become ineffective for random-access queries that seek to retrieve a tuple at an arbitrary rank position $i$, where $i$ may be as large as $O(n)$.

Direct access to conjunctive queries in database systems has been studied to directly access a specific position in the join result, without fully materializing the join~\cite{tziavelis2024ranked, carmeli2023tractable, eldar2023direct, bagan2008computing}.
However, these works assume a {\em prespecified ordering} of the tuples (e.g., lexicographic order on an attribute) and focus on skipping the join operation.
As a result, such approaches are not suitable for settings where the ranking function is part of the input query\footnote{Please refer to Appendix~\ref{sec:app:related} for a complete review of the related work.}.

\stitle{Contributions} In this paper, we formalize and study the {\em Random-Access Ranked Retrieval} (\ith) problem and the {\em Random-Access Similarity Search} problem (\similarity) (Section~\ref{sec:def}). We show a reduction from \similarity\ to \ith\ problem.
Next, we propose an algorithm for \ith\ that uses computational geometric concepts of duality, arrangements, and levels of arrangements.
However, this algorithm suffers from the curse of dimensionality, and its space complexity increases exponentially with dimension (Appendix~\ref{sec:app:kthlevel}). 

Consequently, we propose several algorithms with linear space complexity even for a large number of attributes, based on the concept of $\eps$-samples (Section~\ref{sec:epssampling}).
We start from the two-dimensional case and extend our algorithm to higher dimensions by combining it with the stripe range searching problem. This approach offers an efficient solution for finding a conformal set for an arbitrary rank position.
Efficiently finding the tuples within a {\em narrow} stripe range is fundamental in the development of our solution.
We formulate this as the {\em stripe range retrieval} (\srs) problem, focusing on developing efficient algorithms for the challenging narrow ranges that contain only a few tuples.

Inspired by the Hierarchical Navigable Small World graphs for approximate nearest neighbor search~\cite{malkov2018efficient,dehghankar2025henn}, our practical and scalable algorithm constructs a hierarchical sampling structure for answering \srs, maintaining the neighborhood regions of points in the hierarchy using smallest enclosing balls (Section~\ref{sec:hierarchical}).

Theoretically, we show that our algorithms achieve near-optimal efficiency up to a logarithmic factor (Appendix~\ref{sec:app:opt}).
In addition, we conduct extensive experiments on real-world and synthetic datasets to evaluate the performance of our algorithms in practice and compare them with existing baselines (Section~\ref{sec:exp}).
Our experiments confirm the \emph{efficiency and scalability} of our proposed algorithms. In particular, the $\eps$-sampling algorithm, when combined with the hierarchical structure, achieves the best scalability with respect to both dimension and dataset size in terms of query time and index size, whereas the baseline methods are mostly affected by the curse of dimensionality.

Due to space limitations, we defer extended discussions on applications, related work, more general settings (including dynamic and non-linear scoring functions), additional experiments, and omitted proofs and pseudocode to the appendix.

\vspace{-2mm}
\section{Applications Overview}
\label{sec:applications}

Random-access ranked retrieval (\ith) arises naturally across a broad range of modern data systems that rely on query-dependent ranking functions. While Examples~\ref{eg-1} and~\ref{eg-3} illustrate its relevance in web and e-commerce applications, similar access patterns appear in databases, data analytics platforms, and decision-support systems. Below, we briefly summarize several representative application domains; a detailed discussion with concrete examples is provided in Appendix~\ref{sec:app:applications}.
Together, these applications highlight \ith as a general-purpose access primitive with broad impact across ranked data management and analysis.

\noindent{\bf Extending Random Access to Ranked Retrieval.}
Random access to arbitrary positions in arrays and sorted lists is a fundamental algorithmic primitive, enabling classic techniques such as {\bf binary search}. In contrast, ranked retrieval systems typically only support sequential access, where retrieving the item at rank $i$ requires enumerating all higher-ranked results. \ith bridges this gap by enabling direct access to arbitrary rank positions under a ranking function, thereby extending a broad class of random-access–based algorithms to query-dependent, high-dimensional ranking settings.

\noindent{\bf Analytical Insights and Decision Support.}
Many analytics and business decision-making tasks require understanding how key metrics behave across different rank positions or percentiles, rather than focusing solely on top-$k$ results. Examples include boundary analysis in targeted advertising, percentile-based customer analytics, and sensitivity analysis around decision cut-offs~\cite{cormode2021relative, shrivastava2004medians, zhu2025approximation, greenwald2001space}. Such tasks require efficient access to arbitrary quantiles and narrow rank intervals, which \ith supports without exhaustive scanning.

\noindent{\bf Database Queries over Derived Attributes.}
Modern database and spreadsheet systems frequently operate over \emph{derived attributes} computed at query time, such as weighted combinations or transformations of existing columns~\cite{levandoski2010preference, ghosh2022jenner, rebman2023industry,ma2024spreadsheetbench}. Since these attributes do not exist at indexing time, supporting efficient ranked access over them is challenging. By treating derived attributes as query-defined scoring functions, \ith enables scalable access to arbitrary rank positions and rank ranges without materializing or sorting the full result set.

\noindent{\bf Individually Fair Decision Boundaries.}
In score-based decision systems, cut-off points often serve as decision boundaries and must be chosen carefully to ensure interpretability and individual fairness~\cite{ahn2019fairsight, liu2024rethinking, balciouglu2022notion, meng2024ranked}. Identifying meaningful boundaries requires localized inspection of ranked items around candidate cut-off points to detect significant score gaps. \ith enables such narrow, rank-local exploration, supporting principled and responsible boundary selection.

\vspace{-2mm}
\section{Formal Definitions}\label{sec:def}


We consider a dataset $\data = \{\point_1, \point_2, \cdots, \point_n\}$ consisting of \(n\) points in \(\Reals^d\), where each point \(\point \in \data\) has \(d\) real-valued (scoring) attributes, aka. features, used for ranking\footnote{Besides scoring attributes, the dataset may contain other attributes such as non-ordinal descriptive attributes that are not used for ranking. In this paper, we use the terms attribute and feature interchangeably to refer to the scoring attributes.}. 
A \emph{scoring function} \(\score: \Reals^d \rightarrow \Reals\) assigns a real-valued score to each point in the dataset\footnote{We use the terms scoring function and ranking function interchangeably.}. We consider the \emph{linear scoring functions}, where the score is defined as a weighted sum of the features. Specifically, given a weight vector \(\ef \in \Reals^d\), the score assigned to a point \(\point\) is computed as the dot product\footnote{We assume that the scoring vector \(\ef\) is normalized, i.e., \(\|\ef\| = 1\).}:
$\score_\ef(\point) = \ef^\top \point = \sum_{j=1}^d \ef[j]\point[j]$.
We may directly call $\ef$ as the scoring function.
Figure~\ref{fig:eps-sampling} (left) provides an example of a dataset $\data$ with $n=16$ points and two attributes $x$ and $y$.
Let the red origin-anchored ray in the figure represent $f$.
The perpendicular projection of each point \(\point\) on the vector $f$ shows $\score_\ef(\point)$. 



Our work extends beyond linear functions, supporting (a) scores expressed as linear combinations of non-linear terms (Appendix~\ref{app:nonlinear}) and (b) common similarity measures, like Euclidean and Cosine similarity (Section~\ref{sec:3:ras}). 

\vspace{-3mm}
\subsection{Random-Access Ranked Retrieval (\ith)}\label{sec:def:DA}

\begin{figure}[!t]
    \centering
    \vspace{-4mm}
    \includegraphics[width=0.99\linewidth]{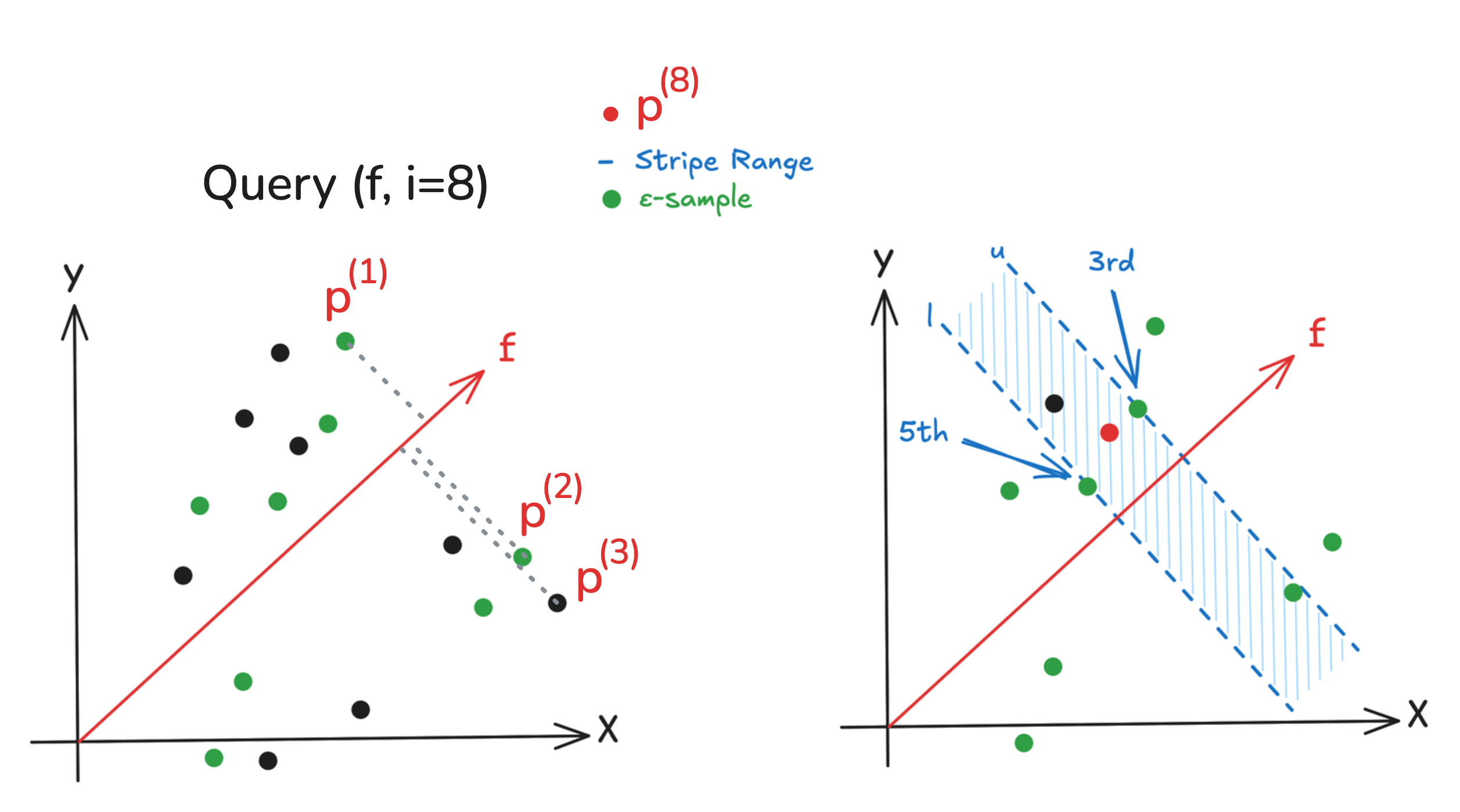}
    \vspace{-6mm}
    \caption{A visual representation of \textsc{Eps2D} algorithm. The toy dataset contain 16 points and the query asks for $(i = 8)$-th point with the scoring function shown in red. (Left) First, we find the bounds by solving \ith~ on the $\eps$-sample. (Right) Finding the $\ell$ and $u$ values, we solve the \srs~ problem on the stripe range.}
    \label{fig:eps-sampling}
    \vspace{-4mm}
\end{figure}

Given a scoring function $\ef$, we define a total order over the dataset \(\data\) by sorting the points in descending order of their scores. 
For example, in Figure~\ref{fig:eps-sampling}, the ordering of the projections on the vector of $f$ shows their ranking.
Let \(\rank_{\data, \ef} : \data \rightarrow \Naturals\) denote the function that assigns to each point in \(\data\) its rank in this ordering (i.e., \(\rank_{\data, \ef}(p) = i\) if \(p\) is the \(i\)-th highest-scoring point under \(\score_\ef\)). When the dataset \(\data\) is clear from the context, we write \(\rank_\ef\) for brevity.
We denote by \(\point^{(i)}\) the $i$-th point (aka the $i$-th element) under \(\score_\ef\), that is,
\(
\point^{(i)} = \rank^{-1}_\ef(i), \quad \forall\; 1 \leq i \leq n.
\) 
The main problem is to efficiently find the point at a specific rank $i$.

\vspace{-2mm}
\begin{pbox}
\begin{problem}[$\kappa$-Random-Access Ranked Retrieval ($\kappa$-\ith)]\label{prob:ith}
    Given a dataset \(\data\) available at preprocessing time, a query consists of a linear scoring function \(\ef\in\Reals^d\) and an integer \(i\in[n]\), find a conformal set $\conf(i) \subset \data$, where $|\conf(i)| = \kappa$ and $p^{(i)} \in \conf(i)$.
\end{problem}
\end{pbox}
\vspace{-1mm}
Based on this definition, when $\kappa = 1$, the task reduces to identifying the exact $i$-th ranked element. As shown in Appendix~\ref{sec:app:opt}, finding $\point^{(i)}$ in high-dimensional spaces is challenging, as its complexity is tied to that of the \emph{half-space range counting} problem~\cite{chazelle1989lower, agarwal2017range}.
For $1 < \kappa \ll n$, this formulation naturally leads to a relaxed variant of the problem: instead of returning a single point, the output is a small set that is guaranteed to contain $\point^{(i)}$. Inspired by the notion of conformal prediction in machine learning~\cite{angelopoulos2021gentle, shafer2008tutorial}, we refer to this set as the \emph{conformal set} for rank~$i$\footnote{\label{fn:kappa}We remove $\kappa$ when it is clear from the context or when we refer to the general problem for any $\kappa$.}.




\vspace{-2mm}
\subsection{Stripe Range Retrieval (\srs)}\label{sec:def:stripe}

While addressing Problem~\ref{prob:ith}, we also solve a side problem that finds all points that fall in a stripe range. Formally, a \emph{stripe range} \(\Stripe_{\ef, \ell, u}\) is defined by a scoring function \(\ef \in \Reals^d\) and two real-valued boundaries \(\ell, u \in \Reals\) such that \(\ell \leq u\). It corresponds to the universe of all points \(x \in \Reals^d\) whose scores based on \(\ef\) lies within the interval \([\ell, u]\), i.e.,

{\small
\vspace{-3mm}
\[
    \Stripe_{\ef, \ell, u} = \{ x \in \Reals^d \mid \ell \leq \score_\ef(x) \leq u \}.
\]
\vspace{-4mm}
}

Given a stripe range query, we aim to return the set of points in $\dee$ that fall inside the given stripe range.
For example, in Figure~\ref{fig:eps-sampling}, the stripe range query is highlighted between the two blue lines.

\vspace{-1mm}
\begin{pbox}
\begin{problem}[Stripe Range Retrieval (\srs)]\label{prob:stripe}
    Given a dataset \(\data\) and a query consisting of a \emph{stripe range}
    \(\Stripe_{\ef, \ell, u}\), return all points in \(\data\) that lie within this range, i.e.,
    \(
    \result = \Stripe_{\ef, \ell, u} \cap \data
    \).
\end{problem}
\end{pbox}

\vspace{-3mm}
\subsection{Random-Access Similarity Search (\similarity)}\label{sec:3:ras}

Similar to the ranked retrieval problem, we define a \emph{random-access variant of the similarity search problem} (also known as nearest neighbor search) under a given similarity metric.

Formally, let
$\dist : \data \times \data \rightarrow \mathbb{R}$
be a distance function that induces the similarity metric on the dataset $\data$. For a query point $q$, we define a ranking function
$\nnrank_q : \data \rightarrow \mathbb{N},$
which orders points in $\data$ by decreasing similarity to $q$ (equivalently, by increasing distance $\dist(q,\cdot)$). In other words, the most similar point to $q$ (i.e., the nearest neighbor) has rank $1$.

Despite the classical NN problem that retrieves the $k$-nearest neighbors for a constant $k$, we consider a more general \emph{random-access} problem, retrieving the point in rank $i$ for any $i \in [n]$.

\vspace{-2mm}
\begin{pbox}
    \begin{problem}[$\kappa$-Random-Access Similarity Search ($\kappa$-\similarity)]\label{prob:similarity}
        Given a dataset $\data$ and a similarity metric $\dist$ available at preprocesing time, a query consists of a point $q \in \mathbb{R}^d$, and an integer $i \in [n]$, find a conformal set $\conf(i) \subset \data$, where $|\conf(i)| = \kappa$ and $p^{(i)} = \nnrank_q^{-1}(i) \in \conf(i)$.
    \end{problem}
\end{pbox}

Similar to Problem~\ref{prob:ith}, when $\kappa = 1$, we look for the exact $i$-th point, while for $1 < \kappa \ll n$, we have a relaxed version of the problem\footref{fn:kappa}.

In this work, we focus on two widely used distance metrics in ANN search and practical similarity search systems: the $\ell_2-$norm (Euclidean distance) and cosine similarity.

In the following, we first show how the \similarity\ problem can be reduced to \ith. Subsequently, we focus on solving \ith, which represents the more general problem.

\color{black}
\vspace{-1mm}
\subsection{Reducing RAS to RAR}\label{sec:reduction}
\begin{figure}[t]
    \centering
    \includegraphics[width=\linewidth]{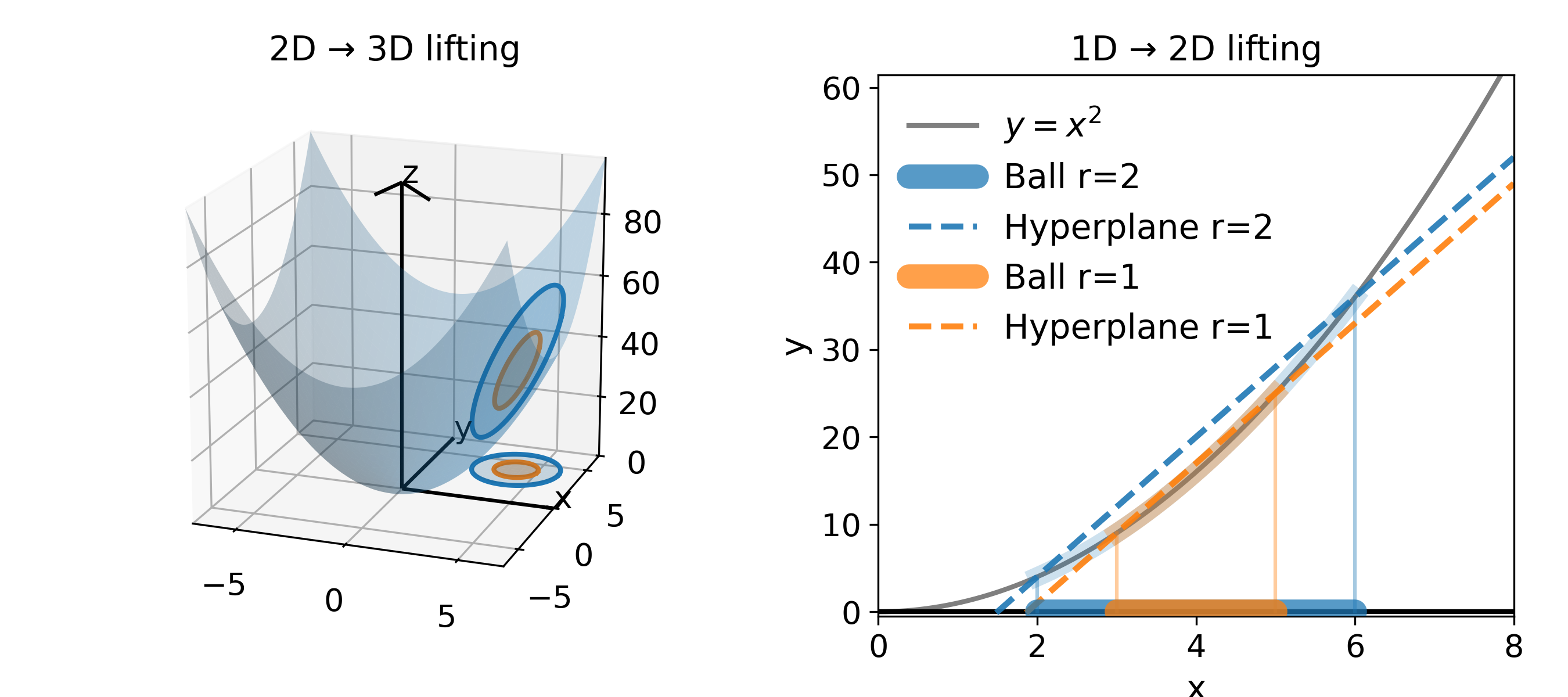}
    \vspace{-8mm}
    \caption{Paraboloid lifting~\cite{har2011geometric}. (Right) Lifting the points in $\mathbb{R}$ into $\mathbb{R}^2$, which transforms segments (1D Balls) into the intersection of half-planes and parabola ($y = x^2$). (Left) Similarly, the $\mathbb{R}^2$ to $\mathbb{R}^3$ lifting.}
    \label{fig:lifting}
    \vspace{-5mm}
\end{figure}

In this section, we use existing techniques to reduce the \similarity\ problem to the \ith\ problem, and show that an efficient solution to the \ith\ problem would also solve the \similarity\ problem. 

\smallskip
\noindent\textbf{Euclidean distance}: Given a query point $q\in \mathbb{R}^d$, 
the $\ell_2-$norm similarity of a point $p\in\mathcal{D}$ to $q$ is measured as
$\dist(p,q)=\|p - q\|_2.$
In Appendix~\ref{app:reduction-1}, we prove Lemma~\ref{lem:reduction-1}, which provides the reduction of the \similarity\ problem to the \ith problem for $\ell_2$-norm:

\vspace{-1mm}
\begin{lemma}\label{lem:reduction-1}
Let $\phi:\mathbb{R}^d \rightarrow \mathbb{R}^{d+1}$ denote the paraboloid lifting transform~\cite{har2011geometric}, which maps a point in $\mathbb{R}^d$ to $\mathbb{R}^{d+1}$ (see Figure~\ref{fig:lifting}). The transformation is defined as
{\small\[
\phi(p) = (p[1], p[2], \ldots, p[d], \|p\|_2^2).
\]}

Using this mapping, Euclidean balls in $\mathbb{R}^d$ correspond to halfspaces in $\mathbb{R}^{d+1}$. Moreover, using the scoring function
{\small\[
f = (-2q[1], -2q[2], \ldots, -2q[d], 1),
\]}
the \similarity\ problem on query $q$ reduces to the \ith\ problem on $f$.
\end{lemma}

\vspace{-1mm}
\smallskip
\noindent\textbf{Cosine distance}: Cosine distance is a dissimilarity measure derived from cosine similarity (more specifically, $1-$ cosine similarity) that captures the angular separation between two vectors. For a query $q$, the cosine distance is given by
$\dist(p,q)=1-({q^\top p})/{(\|q\|_2\cdot\|p\|_2)}.$
Lemma~\ref{lem:reduction-2}, proven in Appendix~\ref{app:reduction-2}, provides the reduction of the \similarity\ problem to the \ith problem for this distance:

\begin{lemma}\label{lem:reduction-2}
    Let $\phi:\mathbb{R}^d \rightarrow \mathbb{R}^d$ denote the normalization map, which scales each point to unit $\ell_2$ norm. The transformation is defined as
    {\small
    \[
    \phi(p) = \frac{p}{\|p\|_2}.
    \]
    }
    Using this mapping, the \similarity\ problem with query point $q$, can be reduced to \ith\ problem with the scoring function:
    {\small
    \[
        f=(\frac{q[1]}{\|q\|_2},\ldots, \frac{q[d]}{\|q\|_2}).
    \]}
\end{lemma}




Note that these reductions equally apply to any value of $\kappa \geq 1$. These transformations, map a ball in $\kappa$-\similarity\ problem to a half-space in $\kappa$-\ith\ problem (see Figure~\ref{fig:lifting}). As a result,
from now on, we focus only on the general $\kappa$-\ith\ problem.

\vspace{-2mm}
\section{Background}\label{sec:background}
In this section, we summarize the key concepts and definitions used in this paper. We assume to be  given a range space $(\data, \mathcal{R})$, where $\data$ is the dataset, and $\mathcal{R}$ is a collection of geometric ranges defined over $\data$, such as all possible balls, half-spaces, or other ranges.

\stitle{VC-Dimension}
The Vapnik-Chervonenkis (VC) dimension~\cite{vapnik2015uniform} for a range space is the size of the largest subset $C \subseteq \data$ such that every subset of $C$ can be realized as the intersection of $C$ with some range in $\mathcal{R}$. A low VC-dimension implies a limited expressive power, which in turn allows for efficient sampling-based approximations~\cite{har2011geometric}.

\stitle{$\eps$-Samples}
An $\eps$-sample of a range space $(\data, \mathcal{R})$ is a subset $\epssample \subseteq \data$ such that, for every range $R \in \mathcal{R}$, the proportion of points in $R$ is approximately preserved in $\epssample$. More formally, $\epssample$ is an $\eps$-sample if for all $R \in \mathcal{R}$,

\vspace{-4mm}
\begin{small}
\begin{align*}
\left| \frac{|R \cap \data|}{|\data|} - \frac{|R \cap \epssample|}{|\epssample|} \right| \le \eps.
\end{align*}
\vspace{-3mm}
\end{small}

For more details on $\eps$-sampling, we refer the reader to \cite{har2011geometric}. 
We use the following foundational theorem on $\eps$-samples as a key tool in deriving our theoretical guarantees:

\begin{theorem}[$\eps$-sample theorem~\cite{vapnik2015uniform}]\label{thm:epssample}
Let $(\data, \mathcal{R})$ be a range space of VC-dimension $\delta$. Then for any $\eps, \varphi \in (0,1)$, a random sample $S \subseteq \data$ of size
\begin{small}
\(
    O\left( \frac{\delta}{\eps^2}\log \frac{\delta}{\eps} + \frac{\delta}{\eps^2}\log \frac{\delta}{\varphi}\right)
\)
\end{small}
is an $\eps$-sample of $(\data, \mathcal{R})$ with probability at least $1 - \varphi$.\hfill\qed
\end{theorem}

In addition to random sampling, which provides a probabilistic method for constructing $\eps$-samples, there also exist deterministic construction algorithms based on discrepancy theory~\cite{chazelle2000discrepancy, har2011geometric}. 

\renewcommand{\arraystretch}{1.4}  
\begin{table*}[!tb]
\centering
\small
\begin{tabular}{|l|l|c|c|c|c|}
\hline
\textbf{Algorithm Name} & \textbf{Summary} & \textbf{Dimension} & \textbf{Space Complexity} & \textbf{Query Time} & \textbf{Problem} \\
\hline
\textsc{KthLevel} 
& Following the levels of arrangement 
& $d \geq 2$
& $O(n^d)$  
& $O(\log n)$
& $1$-\ith \\
\hline
\textsc{Eps2D} & $\eps$-sampling in 2D & $d=2$ & $O(n)$ & $O(n^{2/3} \log n)$ & $1$-\ith \\
\hline
\multirow{2}{*}{\textsc{EpsRange}} & \multirow{2}{*}{$\eps$-sampling + Range Searching} & \multirow{2}{*}{$d \geq 2$} & \multirow{2}{*}{$O(n)$} & $O(\max\{n^{1 - 1/d}, \frac{1}{\eps^2}\log\frac{1}{\eps}\})$ & $(\eps n)$-\ith \\
 &  &  &  & $O(\max\{n^{1 - 1/d}, n^{2/3}\log n\})$ & $1$-\ith \\
\hline
\textsc{EpsHier} & $\eps$-sampling + Hierarchical Sampling & $d \geq 2$ & $O(n)$ & Worst-case $O(n)$; Practically Efficient & $(\eps n)$-\ith \\
\hline
\end{tabular}
\caption{Overview of the proposed methods for the \ith~ problem.\footnotemark}
\label{tab:solution-overview}
\vspace{-9mm}
\end{table*}

\vspace{-2mm}
\section{Solution Overview}\label{sec:overview}

In this section, we provide a high-level overview of our proposed algorithms and summarize their theoretical guarantees. 

The baseline approach for answering the \ith queries uses the extension of the classic divide and conquer (D\&C) median finding algorithm, which, instead of the median, finds the element at position $i$~\cite{blum1973time,cormen2022introduction}.
To do so, it first makes a linear pass over $\data$, and for each point $p$, computes \(\score_\ef(p)\). Then, using the D\&C algorithm, it finds \(\rank^{-1}_{\ef}(i)\) in $O(n)$. The Quick Select algorithm ~\footnote{\href{https://en.wikipedia.org/wiki/Quickselect}{https://en.wikipedia.org/wiki/Quickselect}} also provides a randomized version of this baseline. 
Instead, our algorithms, summarized in Table~\ref{tab:solution-overview}, provide more efficient solutions by preprocessing the data and constructing proper data structures:

\stitle{\textsc{KthLevel}}
Our first algorithm for solving the \ith~problem uses the computational geometric concepts of {\em duality} and {\em arrangement} of hyperplanes~\cite{edelsbrunner1987algorithms, chazelle1985power}.
This algorithm, named \textsc{KthLevel}, keeps track of different {\em levels of the arrangement}, which are searched at query time to efficiently answer \ith queries, i.e., finding the tuple at a given rank position $i$.
Due to the space constraints, further details about this algorithm and its analysis are presented in Appendix~\ref{sec:app:kthlevel}.

This method is
an exact solution to the 1-\ith problem ($\kappa = 1$) that offers
highly efficient {\em logarithmic query time}, but it is not practical for high-dimensional settings ($d > 2$) as it has a {\em space complexity exponential to the number of attributes $d$}.
Our next algorithms based on $\eps$-samples aim to maintain linear space complexity.

\stitle{$\eps$-sampling} We start with the simple two-dimensional case, also illustrated in Figure~\ref{fig:eps-sampling}. During the preprocessing, we compute an \(\eps\)-sample of the dataset, denoted by \(\epssample_\eps\). Given a scoring function \(\ef\) and a target rank \(i\), we approximately solve the \ith~problem on the smaller set \(\epssample_\eps\), rather than the full dataset \(\data\).
From this approximate solution, we derive a stripe range containing the true \(i\)-th ranked point. By solving an instance of stripe range searching (\srs) we find this exact point. We refer to this algorithm as \textsc{Eps2D}. A visual representation of this approach is shown in Figure~\ref{fig:eps-sampling}.

For higher-dimensional settings, we follow a similar generalized strategy (\textsc{EpsRange}). However, unlike in 2D, standard range searching algorithms do not perform well in high-dimensional datasets. To overcome this, we propose a practical algorithm to solve the intermediate \srs problem with hierarchical sampling (\textsc{EpsHier}).

\stitle{Lower Bounds and Optimality}
A discussion on the lower bounds and optimality of our algorithms is provided in the Appendix~\ref{sec:app:opt}.
\vspace{-3mm}
\section{$\eps$-sampling Approach}\label{sec:epssampling}
In theory, the solution based on constructing the arrangement of dual hyperplanes provides a logarithmic time for answering the $\ith$ queries (see Appendix~\ref{sec:app:kthlevel}). However, it is not practical, especially when the number of ranking attributes is not small ($d \geq 3$), since the size of the arrangement (hence, the space complexity) increases exponentially with $d$.
Therefore, in this section, we propose practical solutions that keep the space complexity linear, i.e., $O(n)$.
Our algorithms are based on the computational geometric concept of $\eps$-samples~\cite{har2011geometric} (also introduced in Section~\ref{sec:background}). 

In the following, 
We start by discussing the two-dimensional \textsc{Eps2D} algorithm, then we will generalize our solution to the higher dimensions, and propose practical methods to solve the intermediate \srs~ problem (Problem~\ref{prob:stripe}).

\vspace{-3mm}
\subsection{\textsc{Eps2D}: $\eps$-sampling in 2D}\label{sec:eps2d}
In this section, we discuss our solution based on $\eps$-samples in 2D, called \textsc{Eps2D}, also illustrated in Figure~\ref{fig:eps-sampling}.

\footnotetext{The Big-Oh notations in our guarantees hide constants that are polynomially dependent on the dimension $d$. Consistent with the prior literature~\cite{agarwal2017range}, our focus is on the exponential dependence on $d$.}

\stitle{Preprocessing}
During the preprocessing phase, the algorithm computes an \(\eps\)-sample of the dataset \(\data\), where the ranges are \emph{stripe ranges}. We denote this $\eps$-sample by \(\epssample_\eps\), for a given parameter \(\eps\).

\stitle{Query Answering}
Let \(m = |\epssample_\eps|\) denote the size of the $\eps$-sample built during the preprocessing, where \(m \ll n\)\footnote{This comes from Theorem~\ref{thm:epssample}, while noting that the VC-dimension of stripe ranges is $\delta=d+1$.}. We take the following steps to answer the query:

(\emph{Step i: Thresholding}) 
During the \emph{query phase}, given a query pair \((\ef, i)\), we begin by solving the $1$-\ith~problem on the \(\eps\)-sample \(\epssample_\eps\). Specifically, we identify the \(i_\ell\)-th and \(i_u\)-th ranked points in \(\epssample_\eps\) with respect to the scoring function \(\ef\), where:
\begin{align}\label{eq:klku}
    i_\ell = \left\lfloor m \left( \frac{i}{n} - \eps \right) \right\rfloor, \quad
    i_u = \left\lceil m \left( \frac{i}{n} + \eps \right) \right\rceil,
\end{align}
and \(m = |\epssample_\eps|\). This step can be done by sorting the points in $\epssample_\eps$ according to $\score_\ef$, or by using linear-time-in-$m$ baseline median-finding algorithms, such as median-of-medians.
These two points serve as score-based thresholds that are guaranteed to bound the score of the \(i\)-th ranked point in the full dataset \(\data\), as stated in the following lemma.

\begin{lemma}\label{lm:boundaries}
Let \(\epssample_\eps\) be an \(\eps\)-sample of the dataset \(\data\) for some \(\eps > 0\). Then, for any linear scoring function \(\ef\), the score of the \(i\)-th ranked point \(\point^{(i)}\) in \(\data\) is bounded as
\vspace{-1mm}
{\small
\[
    \score_\ef(q_\ell) \leq \score_\ef(\point^{(i)}) \leq \score_\ef(q_u),
\]}
where \(q_\ell\) and \(q_u\) are the \(i_\ell\)-th and \(i_u\)-th ranked points in \(\epssample_\eps\) under \(\ef\), i.e.,
\vspace{-2mm}
{\small
\[
    q_j = \rank^{-1}_{\epssample_\eps, \ef}(i_j), \quad \text{for } j \in \{\ell, u\}.
\]}
\end{lemma}

\begin{proof}\vspace{-1mm}
See Appendix~\ref{sec:app:proofs} for proof.
\end{proof}

(\emph{Step ii: \srs}) 
Next, we compute the score thresholds corresponding to the boundary points identified in the previous step:
{\small
\begin{align}\label{eq:threshold}
    \ell = \score_\ef(q_\ell), \quad u = \score_\ef(q_u).
\end{align}
}

Using these values, we define the stripe range \(\Stripe_{\ef, \ell, u}\). We then solve an instance of the \srs~problem on the dataset \(\data\) using this stripe range. The result is a subset of points within the score interval \([\ell, u]\), which we denote by \(\result\). This set contains the \(i\)-th ranked point \(\point^{(i)}\), since $\epssample_\eps$ is an $\eps$-sample.
Note that this set is a conformal set for rank $i$, hence providing a solution to the relaxed $|\data_o|$-\ith~problem.

(Step iii: Range Counting) 
We define the half-space \(H_u\) as,
{\small
\begin{align}
    H_u = \{ \point \in \data \mid \point^\top \ef \geq u \},
\end{align}}
\(H_u\) contains all points in \(\data\) whose score under \(\ef\) is at least \(u\). We then apply a half-space range counting algorithm in two dimensions to compute the number of points lying within \(H_u\)~\cite{matouvsek1991efficient}.

(Step iv: Final Selection) 
In the final step, we sort the points in the candidate set \(\result\) in descending order according to the scoring function \(\score_\ef\). Let \(|H_u|\) denote the number of points in \(\data\) whose score is greater than or equal to \(u\), as computed in the previous step. We then return the \((i - |H_u|)\)-th point in the sorted list as the exact answer to the \ith~query.

A pseudo-code of \textsc{Eps2D} is shown in Algorithm~\ref{alg:eps2d-query} in Appendix~\ref{sec:app:pseudo}.
The analysis of the space and time complexity of \textsc{Eps2D} algorithm is also provided in Appendix~\ref{sec:app:eps2d}. The result can be summarized in the following theorem:

\vspace{-2mm}
\begin{theorem}
    The \textsc{Eps2D} algorithm solves the exact \ith~problem in 2D using linear space in $n$ and \( O(n^{2/3} \log n) \) query time\footnote{Note that we assumed the set $\epssample$ to be an $\eps$-sample. According to Theorem~\ref{thm:epssample}, this happens with high probability via random sampling.}.
\end{theorem}
\subsection{\textsc{EpsRange}: $\eps$-sampling in Higher Dimensions}
Now, we present the generalization of the \textsc{Eps2D} algorithm to higher dimensions (\( d \geq 2 \)), referred to as the \textsc{EpsRange} algorithm. The preprocessing phase remains unchanged: we compute an \(\eps\)-sample \(\epssample_\eps\) of the dataset \(\data\), of size \( m \), for a chosen parameter \(\eps\). We know that $m \ll n$, because the VC-dim of stripe ranges is $\delta=d + 1$~\cite{har2011geometric}.

During query time, the algorithm proceeds similarly to \textsc{Eps2D}. In Step~i, we find the $i_\ell$-th and $i_u$-th points in \(\epssample_\eps\) according to the scoring function \(\ef\). 
In Step~ii, we solve a high-dimensional instance of the \srs~ problem. The resulting set \(\result\) contains candidate points from \(\data\) that lie between the scores of the \(i_\ell\)-th and \(i_u\)-th ranked points in the sample. This set is a conformal set for the relaxed problem $|\data_o|$-\ith.
To obtain the exact solution, we proceed to Step~iii, where we perform a high-dimensional range counting query to determine the exact rank of each point in \(\result\). 
Finally, in Step~iv, we sort the candidates in \(\result\) and return the point with exact rank \(i\), completing the solution to the \ith~ selection problem.

\stitle{Analysis}
The analysis is similar to the 2D case, provided in Appendix~\ref{sec:app:epsrange:analysis}. The results are as follows:

\vspace{-1mm}
\begin{theorem}
    The \textsc{EpsRange} algorithm solves the $\kappa$-\ith~problem with linear space usage and $O(\max\{n^{1-1/d}, \frac{d}{\eps^2}\log \frac{d}{\eps}\})$ time, where $\kappa = O(\eps n)$. The exact version, $1$-\ith, is also solved with the same space usage and $O(\max\{n^{1 - 1/d}, dn^{2/3}\log n\})$ time.
\end{theorem}

\vspace{-2mm}
\subsection{Stripe Range Searching with Hierarchical Sampling}\label{sec:hierarchical}

\begin{figure}[t]
    \centering
    \includegraphics[width=0.9\linewidth]{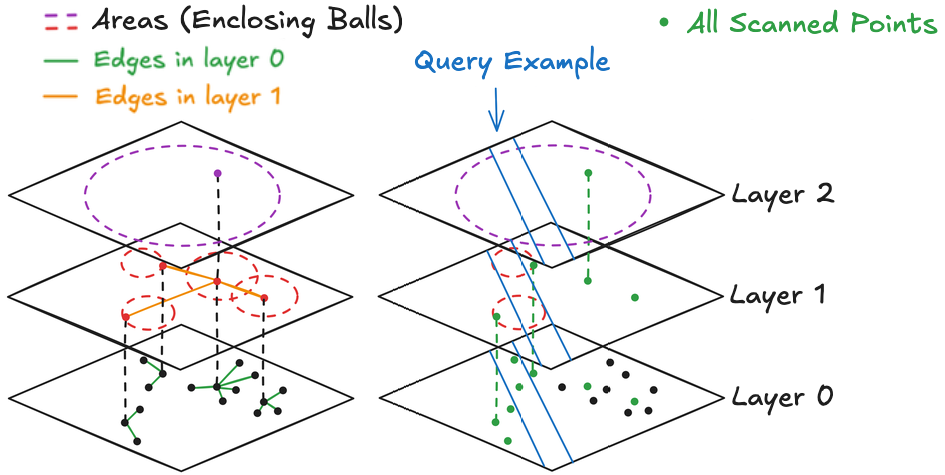}
    \vspace{-4.5mm}
    \caption{A visual representations of the Hierarchical Sampling structure for solving \srs~problem. In this setting, decay rate $r = 4$. (Left) The data structure built during the preprocessing. (Right) An example of running a query on top of this structure, the green points are all the points scanned during the query phase.} 
    \label{fig:hierarchical}\vspace{-4mm}
\end{figure}

The main bottleneck in the runtime of the \textsc{EpsRange} algorithm lies in the range searching subroutines used to solve the \srs~ problem. These subroutines often suffer from the curse of dimensionality, particularly for general simplices, such as non-orthogonal half-spaces within stripe ranges. 
Even algorithms with theoretical sublinear time guarantees tend to perform poorly in practice~\cite{matousek1992reporting}, with a runtime comparable to or worse than a naive linear scan over the entire dataset (see experiments in Section~\ref{sec:exp}).

To address this limitation, we propose a practical algorithm for solving the \srs~ problem, where the ranges are stripes in $d$ dimensions. We empirically show that augmenting \textsc{EpsRange} with this solution leads to significant performance improvements on real-world high-dimensional datasets (\textsc{EpsHier} algorithm).

We present a hierarchical data structure for solving the stripe range searching problem. We begin with a high-level overview of the structure, followed by a description of the \emph{preprocessing step}, and conclude with the details of the \emph{query phase} of the algorithm.

\stitle{Overview}
We observe that the stripe range queries constructed in the solution of the \ith~problem correspond to {\em narrow regions, especially for small choices of~$\eps$}. To exploit this, we construct a hierarchical structure over the input point set that enables efficient pruning of irrelevant points when answering such range queries.

The core idea is to organize the points in a hierarchy that preserves spatial proximity, inspired by similar structures used in nearest neighbor search~\cite{malkov2018efficient, haussler1986epsilon, balltree,dehghankar2025henn}\footnote{See Related Work (Appendix~\ref{sec:app:related}) for comparison.}. Each point in the hierarchy maintains a list of neighbors, and during preprocessing, we compute and store an \emph{enclosing ball} around each such neighborhood\footnote{Not necessarily centered at that point.}.

During query processing, we employ a simple yet effective heuristic: if the stripe intersects the enclosing ball of a node, we explore the node; otherwise, we safely prune it. This significantly reduces the number of points examined during query time.

Empirical results demonstrate that for the narrow stripe ranges used in the \ith~problem, this approach leads to a fast algorithm by eliminating many unnecessary explorations. See Figure~\ref{fig:hierarchical} (right), which illustrates the points explored during an example query.

\subsubsection{Preprocessing} \;

\stitle{Layers}
Given the input dataset $\data$, we construct a hierarchical graph $\mathcal{G}(\data)$ on top of it. The graph consists of $L$ layers built via recursive random sampling.
We define the base layer as the dataset itself, 
$\mathcal{L}_0 = \data$.
Each subsequent layer $\mathcal{L}_\ell$ is formed by randomly sampling $\frac{|\mathcal{L}_{\ell-1}|}{r}$ points from the previous layer $\mathcal{L}_{\ell-1}$, where $r$ is a hyperparameter known as the \emph{exponential decay rate}.
Thus, the size of each layer satisfies:
\begin{small}
\[
    |\mathcal{L}_\ell| = \frac{|\mathcal{L}_{\ell-1}|}{r}, \quad \forall\ 1 \leq \ell \leq L
\]
\end{small}

\stitle{Edges}
Once the layer $\mathcal{L}_\ell$ is sampled, each point in $\mathcal{L}_\ell$ serves as a centroid for partitioning the points in the layer below, $\mathcal{L}_{\ell-1}$. Specifically, each point $p \in \mathcal{L}_{\ell-1}$ is connected to its nearest neighbor in $\mathcal{L}_\ell$, forming a \emph{directed edge} from the centroid $c^*$ to $p$, where
\begin{align}\label{eq:edges}
    c^* = \arg\min_{c \in \mathcal{L}_\ell} ||p - c||_2.
\end{align}
This procedure creates a layered structure, where each level induces a graph by connecting lower-layer points to their nearest centroids in the current layer (see Figure~\ref{fig:hierarchical}, left, where each level illustrates the induced graph).
For every node $\point$ in the graph $\mathcal{G}(\data)$, let $N_\ell(\point)$ denote the set of its neighbors at layer $\mathcal{L}_\ell$. For example, in Figure~\ref{fig:hierarchical} (left), the root node has three red neighbors at layer 1 and five black neighbors at layer 0.

\stitle{Area of Nodes}
For each node $\point$ in layer $\ell$ of the graph $\mathcal{G}(\data)$, we associate a set called the \emph{area} of the node, denoted by $\mathcal{A}_\ell(\point)$. The area captures the set of points from the base layer that are hierarchically covered by this node.
The definition is recursive. For the base layer, we set:
$\mathcal{A}_0(\point) = \{\point\}.$
For any node $\point$ in layer $\ell > 0$, the area is defined as:
{\small\[
    \mathcal{A}_\ell(\point) = \bigcup_{q \in N_{\ell - 1}(\point)} \mathcal{A}_{\ell - 1}(q),
\]}
where $N_{\ell - 1}(\point)$ denotes the neighbors (i.e., children) of $\point$ in layer $\ell - 1$.
In other words, the area of a node at layer $\ell$ is the union of the areas of all its children in the layer below.
As illustrated in Figure~\ref{fig:hierarchical}, we represent the area of each node using the smallest enclosing circle, shown with dotted lines.

\stitle{Enclosing Balls}
For each node $\point$ at layer $\ell$, we define $\Ball_\ell(\point)$ as the \emph{smallest enclosing ball} that covers the area $\mathcal{A}_\ell(\point)$. These balls are later used during query time to enable efficient pruning of nodes when searching within a stripe range:
{\small\[
    \Ball_\ell(\point) = \textsf{Enclosed-Ball}(\mathcal{A}_\ell(\point)).
\]}

\vspace{-1mm}
\stitle{Construction Algorithm}
To construct the hierarchical graph $\mathcal{G}(\data)$, we employ a bottom-up recursive approach. Starting from the base layer, we iteratively sample a subset of points to form the next layer. At each level, we connect nodes in the current layer to their nearest centroids in the layer above, thereby forming directed edges. After establishing the connections, we compute and update the area sets and corresponding enclosing balls for the newly constructed layer. The pseudo-code of the preprocessing procedure is detailed in Algorithm~\ref{alg:hierarchical-preprocess} in Appendix~\ref{sec:app:pseudo}.
Please refer to the Appendix~\ref{sec:app:hierarchy} for the analysis of this algorithm.

\stitle{Querying Algorithm}
Given a stripe range $\Stripe_{f,\ell, u}$, the goal is to report all points in $\Stripe_{f,\ell, u} \cap \data$. The query process begins at the top of the hierarchy, starting from layer $\mathcal{L}_L$.
At each layer $\ell$, we examine whether the enclosing ball $\Ball_\ell(\point)$ of each point $\point \in \mathcal{L}_\ell$ intersects the stripe $\Stripe_{f,\ell, u}$. If there is no intersection, the corresponding node is pruned from further exploration. Otherwise, we proceed to explore its neighbors in the next lower layer, denoted by $N_{\ell - 1}(\point)$.
The full query procedure is provided in Algorithm~\ref{alg:hierarchical-query} in Appendix~\ref{sec:app:pseudo}. A visual illustration of this process is shown in Figure~\ref{fig:hierarchical} (right).

\section{Experiments}\label{sec:exp}
In this section, we evaluate the performance of our algorithms in both the random-access ranked retrieval and similarity search settings.
Our code is publicly available\footnote{\href{https://github.com/UIC-InDeXLab/random-access-search}{\textcolor{blue}{{https://github.com/UIC-InDeXLab/random-access-search}}}}.
We begin by describing the datasets.
Next, we present results for the \emph{stripe range searching} subtask showcasing our proposed hierarchical index (Subsection~\ref{sec:exp:srs}).
We then evaluate our methods on two primary tasks: \emph{random-access ranked retrieval} (Subsection~\ref{sec:exp:ith}) and \emph{random-access similarity search} (Subsection~\ref{sec:exp:similarity}).
Additional results and experimental details are deferred to Appendix~\ref{sec:app:exp}.

\begin{figure}[t]
    \vspace{-1mm}
    \centering
    \includegraphics[width=0.495\linewidth]{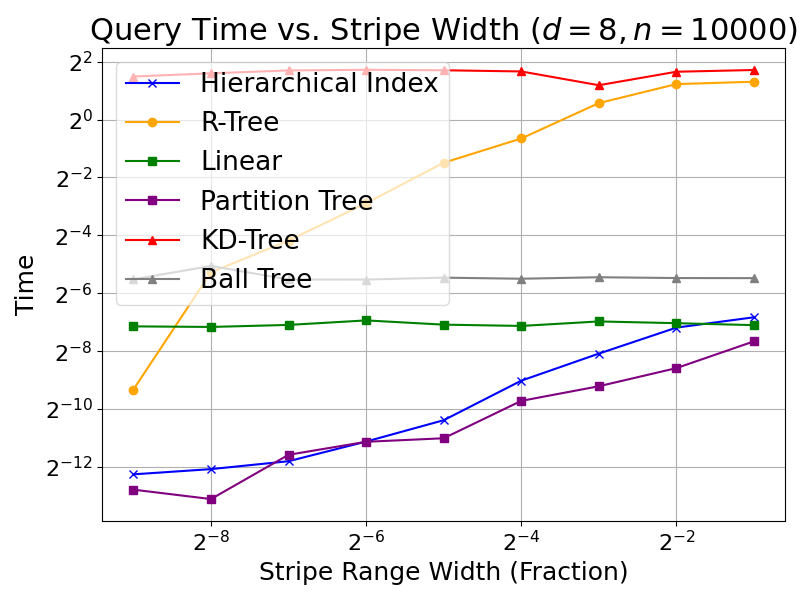}
    \includegraphics[width=0.495\linewidth]{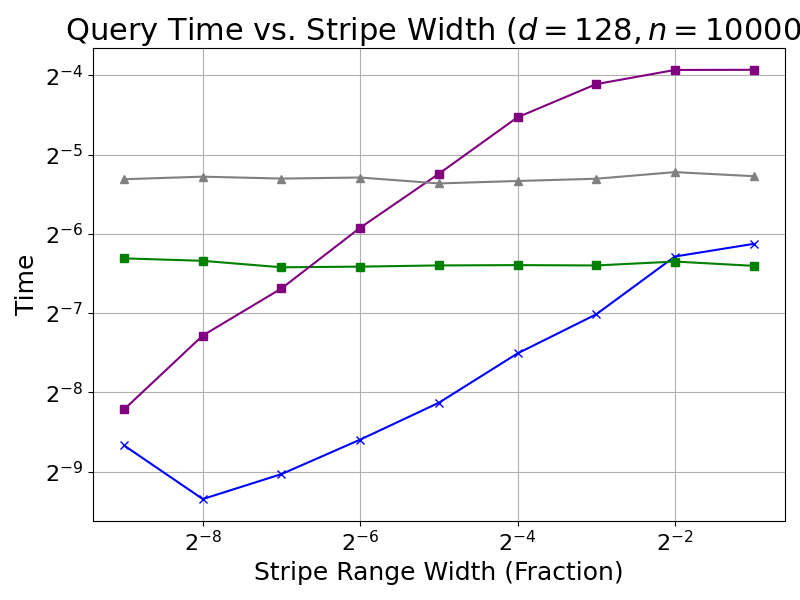}
    \vspace{-7mm}
    \caption{Comparison of \srs query time with respect to the stripe width across different dimensionalities on the synthetic data.}
    \label{fig:srs:time-width}
    \vspace{-1mm}
\end{figure}

\stitle{Datasets}
For the ranked retrieval, we use synthetic and real-world datasets. The synthetic dataset is generated from a Zipfian distribution. Our real-world datasets include \textsc{US Used Cars}~\cite{mital2020uscars} (3M records, 66 columns), \textsc{FIFA 2023}~\cite{leone2022fifa23} (300K records, 54 columns), and \textsc{US Flights}~\cite{usdot_flight_delays} (5M records, 31 columns). 
For the similarity search setting, we evaluate our methods on five popular benchmark datasets from
\href{https://ann-benchmarks.com}{\texttt{ann-benchmarks.com}}:
\textsc{SIFT-128-Euclidean}~\cite{jegou2010product} (1M vectors, $\ell_2$),
\textsc{Fashion-MNIST-784-Euclidean}~\cite{xiao2017online} (60K vectors, $\ell_2$),
\textsc{GIST-960-Euclidean}~\cite{jegou2010product} (1M vectors, $\ell_2$),
\textsc{GloVe-100-Angular}~\cite{pennington2014glove} (1M vectors, cosine), and
\textsc{GloVe-25-Angular}~\cite{pennington2014glove} (1M vectors, cosine)\footnote{The number in the dataset name represents the dimensionality of the vectors.}.
More details of the datasets are provided in Appendix~\ref{sec:app:exp:datasets}.

\vspace{-3mm}
\subsection{Stripe Range Retrieval}\label{sec:exp:srs}
For the \srs problem, we compare against classical and widely used range searching indices, including KD-Tree~\cite{bentley1975multidimensional}, R-Tree~\cite{guttman1984r}, and Partition Tree~\cite{matouvsek1991efficient}, and Ball Trees~\cite{balltree}. Our proposed hierarchical index is abbreviated as \textsf{Hierarchical}. As a baseline, we also include a simple linear search approach, denoted as \textsf{Exhaustive}.
To apply the baseline range searching indices to stripe range queries, we traverse the tree top-down, pruning subtrees whose corresponding hyperrectangles do not intersect the query stripe (see Algorithm~\ref{alg:hierarchical-query}). More experiments on \srs\ problem is provided in Appendix~\ref{sec:app:exp:srs}.

\begin{figure}[t]
    \vspace{-3mm}
    \centering
    \includegraphics[width=0.495\linewidth]{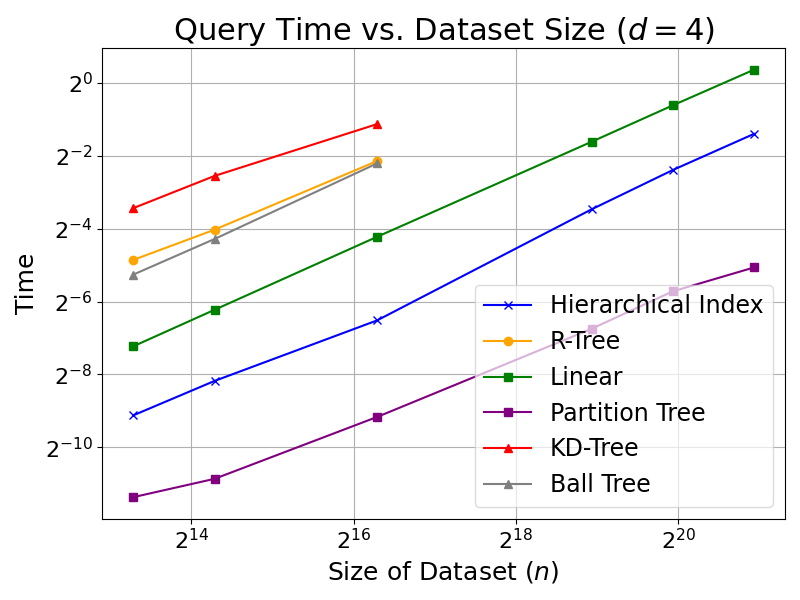}
    \includegraphics[width=0.495\linewidth]{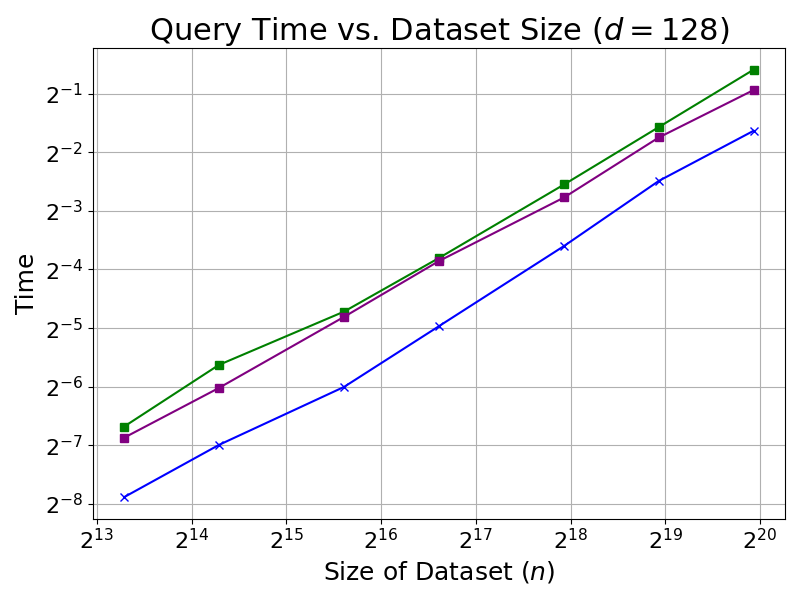}
    \vspace{-7mm}
    \caption{Comparison of the query time of Hierarchical Sampling on \srs\ problem vs dataset size. Large $d$ and $n$ cases do not use \textsf{KD-tree} or \textsf{R-tree}, as these methods do not scale well.}
    \label{fig:srs:time-size}
\end{figure}

\subsubsection{Synthetic Dataset}
We begin by evaluating the performance of the algorithms on synthetic datasets.

\stitle{Effect of Stripe Width}
For these experiments, the scoring function $\ef$ is sampled uniformly from the unit hypersphere in $\mathbb{R}^d$. We vary the stripe $\Stripe_{\ef, \ell, u}$ by adjusting its width, i.e., the number of points it contains. Figure~\ref{fig:srs:time-width} presents the query time performance across different dimensionalities $d$ as a function of stripe width.
We observe that both \textsf{KD-Tree} and \textsf{R-Tree} methods fail to scale effectively to high-dimensional settings, and thus we exclude them from experiments with $d = 128$. The proposed \textsf{Hierarchical Sampling} method achieves {\bf up to a 16$\times$ speedup} over \textsf{Exhaustive} search, particularly in \emph{narrow stripe} queries. Such narrow stripes frequently arise when solving the \ith problem, indicating that this method is also helpful in this context. While the \textsf{Partition Tree} also shows a fast query time in low-dimensional spaces, its performance degrades with increasing dimensionality.

\stitle{Effect of Dataset Size}
Figure~\ref{fig:srs:time-size} illustrates the impact of increasing the dataset size $n$ on the query time of the algorithms, for both low- and high-dimensional settings. As observed, the \textsf{Hierarchical Sampling} method consistently outperforms the baselines.

\stitle{Real Datasets}
 Figure~\ref{fig:srs:real} shows the comparison on \textsc{Us Used Cards} dataset. As shown, \textsf{Hierarchical Sampling} significantly outperforms both \textsf{Exhaustive} search and the \textsf{Partition Tree} method. In this experiment, multiple random scoring functions are sampled uniformly from the unit hypersphere. (See Appendix~\ref{sec:app:exp:srs} for other real-datasets).

\begin{figure}[t]
    \vspace{-3mm}
    \centering
    \includegraphics[width=0.495\linewidth]{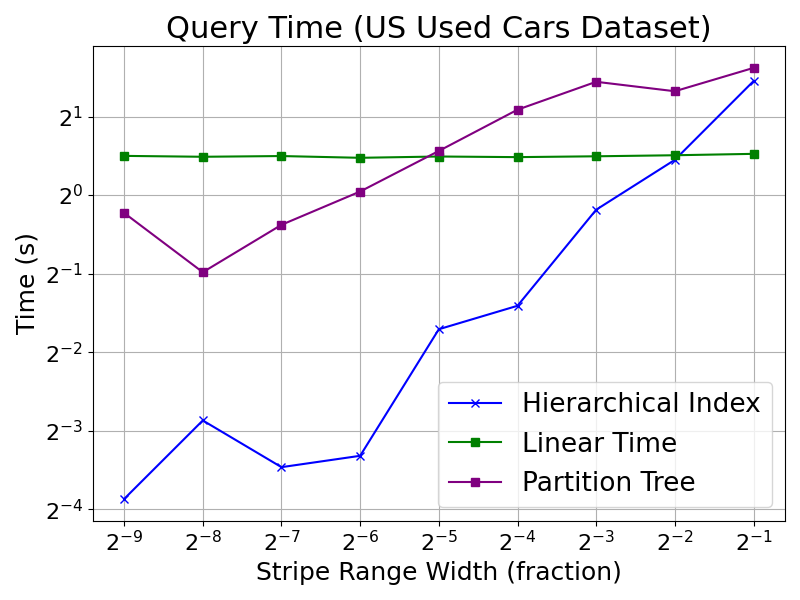}
    \includegraphics[width=0.495\linewidth]{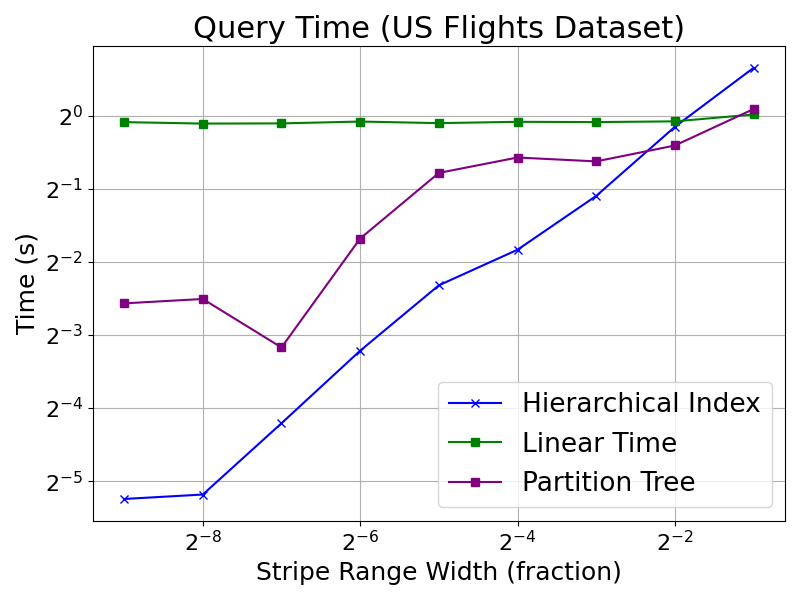}
    \vspace{-7mm}
    \caption{Comparison of the query time of Hierarchical Sampling on \srs\ problem vs the stripe range width.}
    \label{fig:srs:real}\vspace{-2mm}
\end{figure}

\begin{figure}[t]
    \vspace{-2mm}
    \centering
    \includegraphics[width=0.495\linewidth]{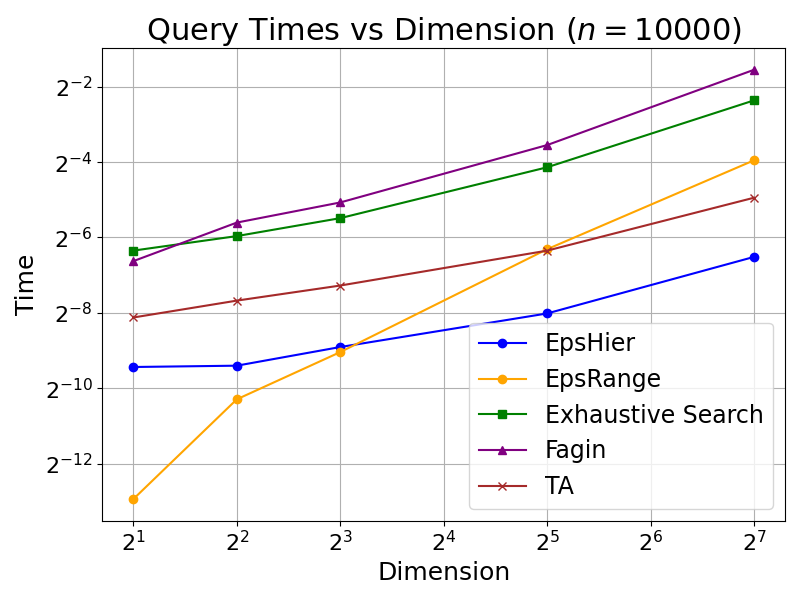}
    \includegraphics[width=0.495\linewidth]{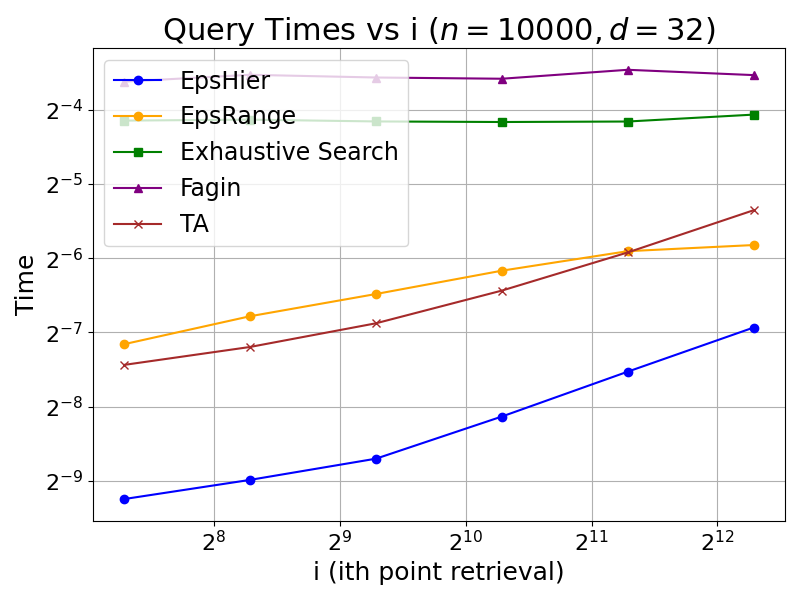}
    \vspace{-7mm}
    \caption{Comparison of the query time of \textsf{EpsRange} and \textsf{EpsHier} on $\kappa$-\ith\ problem vs dimension $d$ and the rank $i$. In this setting, $\eps = \frac{1}{16}$.}
    \label{fig:dar:time-dim-i}\vspace{-6mm}
\end{figure}

\stitle{Accuracy of the Output} In all the settings, the output is guaranteed to have a recall of $100\%$, reporting all the points in the range.

\vspace{-2mm}
\subsection{Random-Access Ranked Retrieval}\label{sec:exp:ith}

We include the Threshold Algorithm (TA) and Fagin's algorithm~\cite{fagin} as baseline methods for retrieving top-K elements\footnote{Please refer to the related work (Appendix~\ref{sec:app:related}).}. Additionally, we consider a simple search baseline, referred to as \textsf{Exhaustive}, which computes all the scores and extracts the $i$-th ranked element. Our proposed algorithms (\textsf{EpsRange} and \textsf{EpsHier}) are used for the $\kappa$-\ith\ problem while the \textsf{KthLevel} algorithm solves the exact $1$-\ith\ problem in 2D (see Appendix~\ref{sec:app:exp:rar}).

\subsubsection{Synthetic Dataset} \;

\vspace{-1mm}
\stitle{Effect of Dimension and value of $i$}
Figure~\ref{fig:dar:time-dim-i} (left) presents the query time of algorithms as a function of dimension $d$. As $d$ increases, both \textsf{Fagin} and \textsf{TA} algorithms exhibit poor scalability for random-access queries. Among the proposed methods, \textsf{EpsHier} consistently outperforms \textsf{EpsRange}, particularly when the $d > 8$. The performance degradation of \textsf{EpsRange} in higher dimensions is because of its reliance on \textsf{Partition Tree} index~\cite{matousek1992reporting}, which suffers from the curse of dimensionality.
Figure~\ref{fig:dar:time-dim-i} (right) illustrates the impact of varying the value of $i$ in the \ith query on the runtime. The \textsf{TA} algorithm shows increasing query time as $i$ grows, since it must retrieve all the top-$i$ items. In contrast, the runtime of \textsf{EpsRange} remains relatively stable, as it depends primarily on the position of the stripe rather than the value of $i$. The \textsf{EpsHier} method also shows a moderate increase in runtime for larger $i$, which can be because of the relative location of the stripes with respect to the dataset. The marginal stripes (smaller values of $i$) intersect with fewer balls in the hierarchical structure proposed in Section~\ref{sec:hierarchical}, and the pruning is more effective in this case.

\stitle{Effect of Dataset Size and $\eps$}
Figure~\ref{fig:dar:time-size-eps} (left) shows the effect of dataset size on the query time of the algorithms. We observe significant speedups achieved by both \textsf{EpsRange} and \textsf{EpsHier} compared to the baseline methods.
Figure~\ref{fig:dar:time-size-eps} (right) illustrates the impact of varying the parameter $\eps$. As expected, this parameter influences the size of the output set, $\kappa$, in the $\kappa$-\ith\ problem. Increasing the value of $\eps$ results in a larger output set, whereas smaller values of $\eps$ require searching over a larger $\eps$-sample during the initial phase of the algorithm. Both extremes lead to increased query times. The most efficient performance is typically achieved for a moderate value of $\eps$, which may serve as a default when not given by user.

\stitle{Real Datasets}
Figure~\ref{fig:dar:flight} presents the results of running the algorithms on the \textsf{US Flights} dataset, evaluating across different values of $\eps$ (left) and different values of $i$ (right). We observe significant speedups achieved by \textsf{EpsHier}, followed by \textsf{EpsRange}. Appendix~\ref{sec:app:exp:rar} shows the results on other datasets.

\begin{figure}[t]
    \vspace{-2mm}
    \centering
    \includegraphics[width=0.495\linewidth]{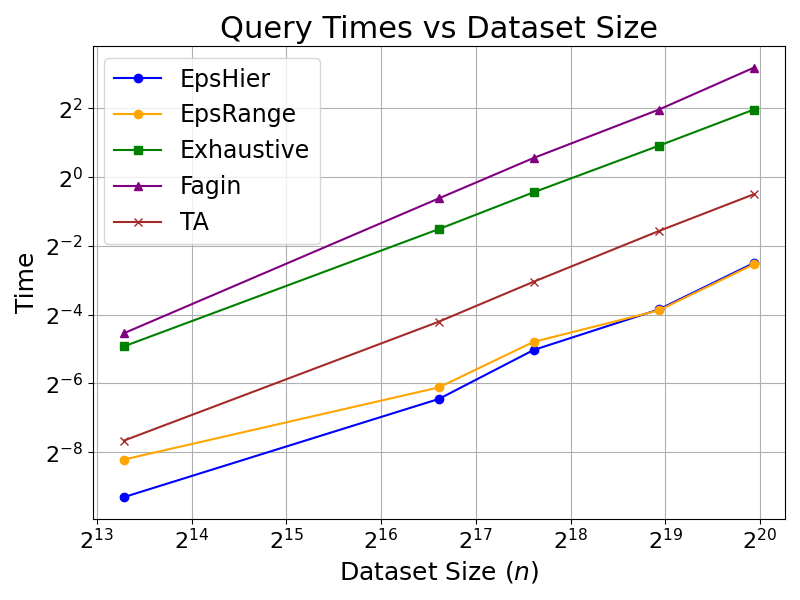}
    \includegraphics[width=0.495\linewidth]{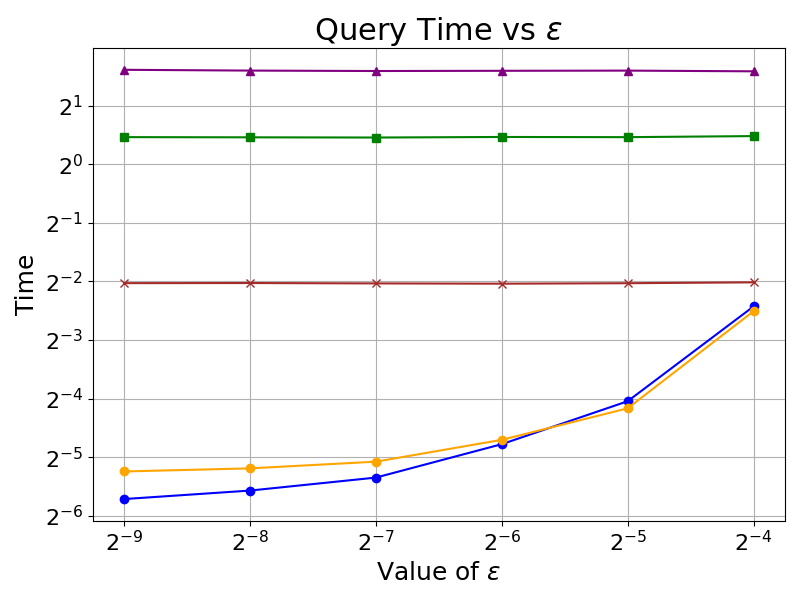}
    \vspace{-4mm}
    \caption{Comparison of the query time of \textsc{EpsRange} and \textsc{EpsHier} on $\kappa$-\ith\ problem vs dataset size and $\eps$. In the right plot, $n = 10000$ and $\kappa = \eps n$.}
    \label{fig:dar:time-size-eps}
\end{figure}

\begin{figure}[t]
    \vspace{-2mm}
    \centering
    \includegraphics[width=0.495\linewidth]{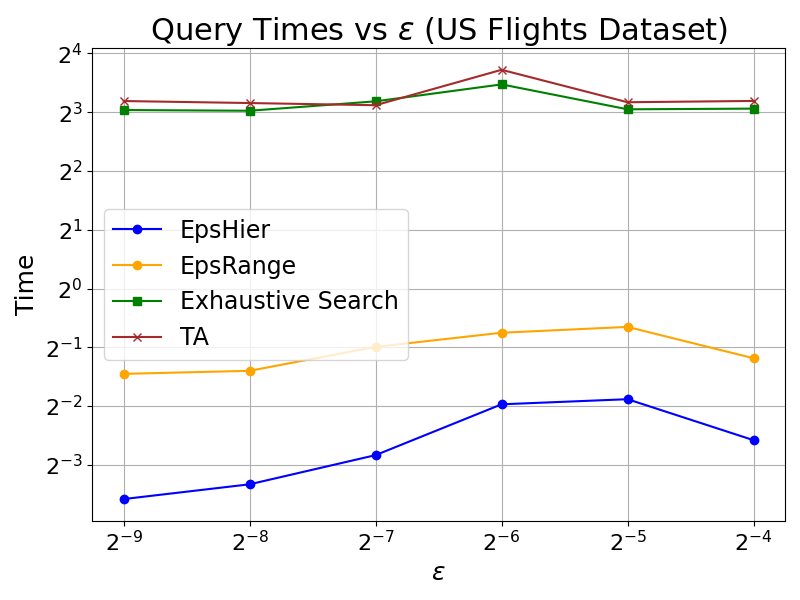}
    \includegraphics[width=0.495\linewidth]{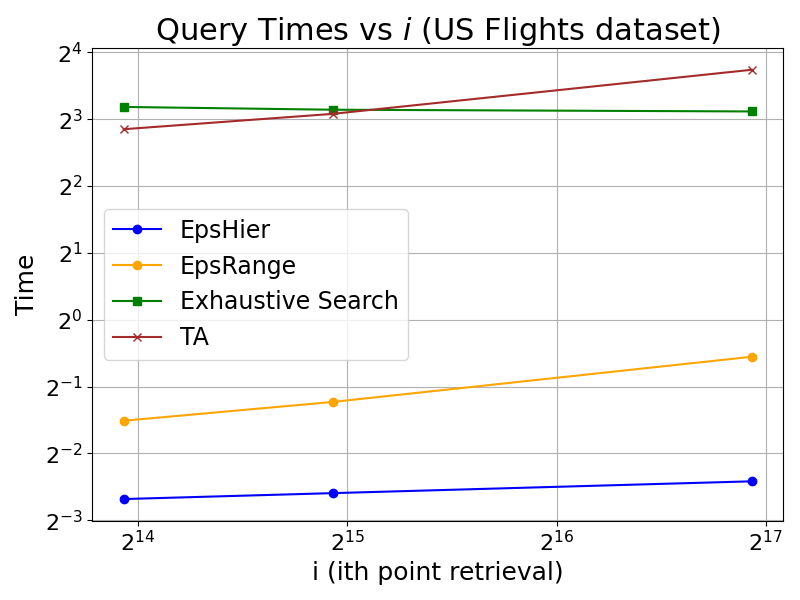}
    \vspace{-7mm}
    \caption{Comparing the query time of \textsc{EpsHier} and \textsc{EpsRange} on $\kappa$-\ith\ problem vs the value of $\eps$ and rank $i$ on US Flights dataset.}
    \label{fig:dar:flight}
    \vspace{-7mm}
\end{figure}

\stitle{Accuracy of the Output} In all the experiments, the returned set is guaranteed to have the $\point^{(i)}$, so the recall is always 100\%.

\vspace{-3mm}
\subsection{Random-Access Similarity Search}\label{sec:exp:similarity}

For the similarity search setting, we first apply the reduction described in Section~\ref{sec:reduction}.
For datasets under the $\ell_2$ metric, we lift the data points during preprocessing and apply the same lifting to the query at query time.
For cosine similarity, we treat the query as a linear scoring function corresponding to the standard dot product.
The objective is to retrieve the $i$-th nearest neighbor of a query point, where $i$ is not necessarily $O(1)$.

We consider six baselines for comparison. Our first baseline is (i) the classical \emph{Quickselect} algorithm for order statistics computation~\footnote{\href{https://en.wikipedia.org/wiki/Quickselect}{https://en.wikipedia.org/wiki/Quickselect}}.
We also evaluate several popular ANN indexes implemented in the widely-used \textsc{Faiss} library~\cite{douze2024faiss}:
(ii) HNSW~\cite{malkov2018efficient},
(iii) PQ~\cite{jegou2010product},
(iv) IVFPQ~\cite{jegou2010product},
(v) IVFFlat~\cite{jegou2010product}, and
(vi) LSH~\cite{guarantee}.
The details of these algorithms are discussed in the related works (Appendix~\ref{sec:app:related}).

In Appendix~\ref{sec:app:exp:similarity}, we describe how these ANN indexes are adapted to support the random-access version of the problem.
We note that, as demonstrated by our experimental results, these methods are generally not well suited for random-access retrieval when $i$ (or equivalently $k$) is large; they are included primarily for completeness and comparison.

Figure~\ref{fig:similarity} shows the recall and query time of our algorithm, \textsc{EpsHier}, compared with the baseline methods on the \textsc{SIFT-128-Euclidean} dataset.
As the target rank $i$ increases to larger values (e.g., the first, second, and third quartiles), all ANN-based methods degrade substantially, exhibiting both increased query time and near-zero recall.
In contrast, \textsc{EpsHier} consistently achieves close to $100\%$ recall while maintaining relatively low query time, even when random-accessing elements in the middle of the ranked order.
This behavior is particularly evident for large ranks $i = O(n)$. Table~\ref{tab:knn-time-recall-comparison} in Appendix~\ref{sec:app:exp:similarity} shows a complete comparison on a wide range of $i$ values and other datasets.

\begin{figure}[t]
    \centering
    \includegraphics[width=\linewidth]{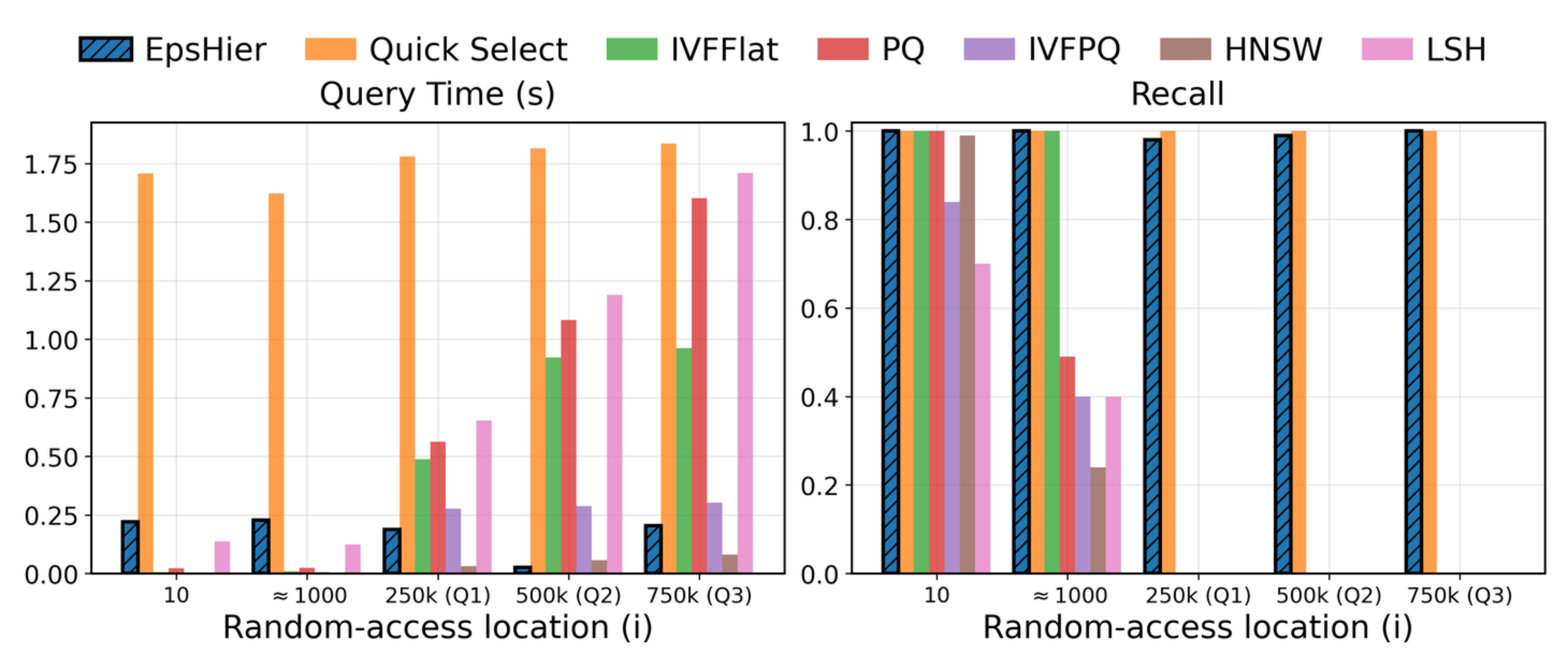}
    \vspace{-7mm}
    \caption{Comparing query time and recall@i on random-access similarity search task. The dataset is \textsc{Sift-128-Euclidean} with 1M vectors of dimensionality 128 with $\ell_2$ distance.}
    \label{fig:similarity}
    \vspace{-3mm}
\end{figure}

\vspace{-3mm}
\section{Conclusion}\label{sec:conclusion}
\vspace{-1mm}
We studied the problem of Random-Access Ranked Retrieval (\ith) and Random-Access Similarity Search (\similarity). We showed that \similarity\ can be reduced to \ith. For solving the general \ith\ problem, we first proposed an algorithm based on geometric arrangements with logarithmic query time but an exponential space complexity.
Next, we proposed our space-efficient algorithms based on $\eps$-samples.
By providing a hierarchical sampling-based data structure for efficient stripe-range retrieval, we developed a practical algorithm for the relaxed $\kappa$-\ith\ problem that scales to hundreds of dimensions.

\balance

\bibliographystyle{ACM-Reference-Format}
\bibliography{ref}

\clearpage
\appendix
\section*{Appendix}
\subsection*{Table of Content}

\noindent
\makebox[\linewidth][l]{1.\quad More on the Application Examples \dotfill\ \ref{sec:app:applications}}\\
\makebox[\linewidth][l]{2.\quad Related Work \dotfill\ \ref{sec:app:related}}\\
\makebox[\linewidth][l]{3.\quad Extension to Non-linear Scoring Functions \dotfill\ \ref{app:nonlinear}}\\
\makebox[\linewidth][l]{4.\quad Missing Details for Reducing RAS to RAR \dotfill\ \ref{app:reduction}}\\
\makebox[\linewidth][l]{5.\quad Our Algorithm Based on $k$-th Level of Arrangement \dotfill\ \ref{sec:app:kthlevel}}\\
\makebox[\linewidth][l]{6.\quad Analysis of the \textsc{Eps2D} Algorithm \dotfill\ \ref{sec:app:eps2d}}\\
\makebox[\linewidth][l]{7.\quad Missing Proofs \dotfill\ \ref{sec:app:proofs}}\\
\makebox[\linewidth][l]{8.\quad Dynamic Setting \dotfill\ \ref{sec:app:dynamic}}\\
\makebox[\linewidth][l]{9.\quad Analysis of the \textsc{EpsRange} Algorithm \dotfill\ \ref{sec:app:epsrange:analysis}}\\
\makebox[\linewidth][l]{10.\quad Analysis of Hierarchical Sampling \dotfill\ \ref{sec:app:hierarchy}}\\
\makebox[\linewidth][l]{11.\quad Pseudo-codes \dotfill\ \ref{sec:app:pseudo}}\\
\makebox[\linewidth][l]{12.\quad Lower Bounds and Optimality of Our Algorithms \dotfill\ \ref{sec:app:opt}}\\
\makebox[\linewidth][l]{13.\quad Extended Experiments \dotfill\ \ref{sec:app:exp}}
\section{More on the Application Examples}\label{sec:app:applications} 
Examples~\ref{eg-1} and~\ref{eg-3} highlight two motivations for random-access ranked retrieval for e-commerce and web applications.
Such examples can be found across a wider range of domains.
In this section, we highlight additional application examples of $\ith$ and its variations across a diverse set of domains.

\stitle{Extending random-access to ranked retrieval}
Random access to arbitrary positions in an array or a sorted list is a fundamental operation underlying many of the most efficient algorithms in computer science. Classic divide-and-conquer techniques, such as {\bf binary search} and {\bf quick select}, critically rely on the ability to randomly access elements at arbitrary positions in order to efficiently prune the search space.

However, in ranked retrieval systems, this basic capability is largely missing. Although items are conceptually ordered by a scoring or ranking function, accessing the element at a given rank typically requires enumerating all higher-ranked items, effectively reducing random access to sequential scan. This limitation prevents the direct application of well-established algorithmic techniques in settings where ordering is induced by a query-dependent ranking rather than a static key.

The random-access ranked retrieval (\ith) problem addresses this gap by providing efficient access to arbitrary rank positions under a ranking function without materializing the entire result prefix. \ith extends the power of random-access–based algorithms to settings where the values are scores, determined by a ranking function or by similarity to a specific point.
This capability opens the door to adapting a broad class of classical algorithms, previously confined to static (sorted) lists, to modern data retrieval workloads where rankings are query-specific (dynamic) and high-dimensional.

\stitle{Analytical Insights for Business Decision Support}
In many business and decision-support applications, stakeholders are interested 
in advanced analytics (e.g., income distribution, demographic distribution, shopping behavior) for 
not only in the top-ranked customers, or the ones in their close vicinity, but in understanding how key analytics look across different rank positions or percentiles of a scoring function~\cite{cormode2021relative, shrivastava2004medians, zhu2025approximation, greenwald2001space}.
Examples include analyzing income or demographic distributions across distance bands from a retail location, studying customer engagement at different relevance thresholds, or evaluating how small changes in targeting criteria affect the selected population.

Consider a targeted advertising scenario in which a business defines its audience using a distance-based or machine-generated ranking function -- for instance, customers ranked by proximity to a store or by predicted likelihood of conversion. While traditional systems efficiently retrieve the top-k customers, business decisions often hinge on understanding boundary cases: customers near a cutoff who are just included or excluded under a given targeting policy. These boundary regions are especially important for evaluating trade-offs, fairness, and return on investment.

Supporting such analyses requires efficient access to records at arbitrary quantiles or rank positions (e.g., “the 90th-percentile customer according to this scoring function”), as well as the ability to explore narrow rank intervals around these positions. Naïve approaches based on sampling or top-k retrieval are ineffective, since points at the same rank or distance may be widely dispersed in the feature space and cannot be approximated by local neighborhoods.

Random-access ranked retrieval enables this form of fine-grained analytical exploration by allowing analysts to directly probe specific rank positions and score bands without exhaustively scanning the dataset. As a result, it provides a powerful primitive for interactive business analytics, sensitivity analysis, and decision support in large-scale, high-dimensional settings.

\stitle{Efficient Database Queries over Derived Attributes}
Modern database systems routinely support queries over {\em derived attributes}, i.e., attributes that are not explicitly stored in the database but are computed on-the-fly as part of a query~\cite{levandoski2010preference, ghosh2022jenner}. Such derived attributes are commonly expressed as functions over existing columns and are treated as first-class citizens in query predicates, projections, and orderings.
For example, consider the following SQL query:

\vspace{2mm}
\begin{verbatim}
SELECT Name, YEAR(Today_Date) - YEAR(Birth_Date) AS Age
FROM Employees
WHERE Age BETWEEN 25 AND 30;
\end{verbatim}\vspace{-4mm}
\vspace{2mm}

Here, {\tt Age} is not a stored attribute but is defined by the query itself as a function of {\tt Birth\_Date}. Conceptually, this corresponds to introducing a new attribute whose value is computed per tuple and then used for filtering or ranking.

From a ranked-retrieval perspective, derived attributes naturally induce query-dependent scoring functions that combine existing attributes through arithmetic or functional transformations. 
Supporting efficient access to tuples based on such derived attributes poses a fundamental challenge. Since derived attributes do not exist at indexing time, traditional indexing techniques cannot be applied directly, and a naïve solution requires computing the derived value for all tuples. This is particularly expensive when users wish to explore arbitrary rank positions or specific rank ranges. 
Also in spreadsheet applications~\cite{rebman2023industry,ma2024spreadsheetbench}, data scientists frequently create new attributes as weighted sums of existing ones. Subsequently, they may {\em ``jump''} to tuples at specific rank positions without explicitly enumerating preceding entries. 

By treating derived attributes as scoring functions defined at query time, \ith\ enables efficient access to arbitrary rank positions without fully materializing or sorting the entire result set. This makes \ith\ a natural foundation for supporting interactive and scalable database queries over derived attributes, extending classical database functionality with modern rank-aware access patterns.

\stitle{Individually Fair Decision Boundaries}
In many decision-making scenarios, the objective is not merely to retrieve items at a fixed rank, but to identify the cut-off points in a ranked list as the decision boundaries~\cite{ahn2019fairsight, liu2024rethinking, balciouglu2022notion, meng2024ranked}. 
To minimize individual unfairness (i.e., to produce different outputs for similar inputs), the cut-off points should indicate meaningful transitions between the classes. In other words, the gap at the boundary should be wide.
A tangible example arises in educational assessment, where instructors must determine grade boundaries (e.g., between an A and a B) based on student scores. Rather than enforcing a predefined percentile or absolute threshold, instructors often examine the score distribution around a target rank (e.g., the 90th percentile) to identify natural breaks that separate performance tiers. 
Similarly, drawing fair decision boundaries in score-based data-driven systems (e.g., for loan applications) requires exploring a band (range) of valid cut-off points to maximize the decision-boundary gap~\cite{ahn2019fairsight}.

Random-access ranked retrieval enables principled gap discovery by supporting direct, localized inspection of ranked items around candidate cut-off points. This capability allows decision makers to efficiently retrieve entities within the decision boundary candidates' range and to set the cut-off point responsibly.


\section{Related Work}\label{sec:app:related}
\stitle{Top-$k$ Retrieval}
Top-$k$ query processing has been extensively studied across various settings. The foundational  Threshold Algorithm (TA)~\cite{fagin} combines sorted and random access to efficiently identify top-$k$ results and has inspired numerous extensions, including cost-based optimizations and early stopping strategies~\cite{taadaptbruno}. Probabilistic top-$k$ methods~\cite{probtopk} adapt this problem to uncertain data, ranking results based on expected scores or statistical confidence. Index-based approaches such as the onion layer~\cite{onion} and view-based techniques~\cite{lpta,rankviews} have been proposed to accelerate queries through structural or materialized reuse. For comprehensive surveys, see~\cite{ilyasrank,gautopk}. 
\underline{\em In top-$k$ retrieval, $k$ is a small constant}, while in our setting, we aim to find the tuple at rank position $i$, where $i$ is $O(n)$. As a result, the algorithms proposed for top-$k$ are inefficient for our setting, as shown in the experiments.

\stitle{Direct Access Queries in Databases}
Recent work in database theory has explored the notion of \emph{direct access} to query answers, aiming to support efficient access to the $i$-th ranked tuple without materializing the full result~\cite{tziavelis2024ranked, eldar2023direct, carmeli2023tractable, bringmann2025tight, bagan2008computing}. In these works, the ordering is generally assumed to be fixed, lexicographically, on some attributes. Some work extends this to conjunctive queries with aggregation, designing algorithms that enable access to top-ranked answers without fully evaluating the query~\cite{eldar2023direct}. 
In our paper, we discuss any ranking function defined as as the weights on the features. Specifically, rather than assuming a predefined ordering, {\em we consider a ranking function as part of} {\em the user query}. 
As a result, the proposed algorithms for direct-access queries cannot be adapted to our setting.

\stitle{Range Searching Indexes}
Range searching, specifically for simplex ranges, is related to our work in solving the intermediate \srs~problem. Some approaches achieve logarithmic query time at the cost of exponential space~\cite{aggarwal1990solving, dobkin1990implicitly}, while others aim for linear space and sublinear query time~\cite{chazelle1985power, haussler1986epsilon, matousek1992reporting, matouvsek1991efficient}. Specialized variants have also been studied, including half-space range searching~\cite{chazelle1985power} and axis-aligned queries~\cite{chazelle1986fractional}. For a comprehensive survey, see~\cite{agarwal2017range}.
To address the \srs~problem in practice, we develop a hierarchical sampling approach inspired by data structures commonly used for similarity search, like HNSW~\cite{malkov2018efficient}. 
As an intermediate problem for solving \ith, \underline{\em our solutions are tailored for narrow ranges}.
While our methods do not offer formal theoretical guarantees, they performed efficiently in practical settings.
Another indexing technique related to our algorithm is the \emph{Ball Tree} structure~\cite{balltree}. 
The key difference between Ball Trees and our approach lies in the construction and structure. 
In our method, each layer is generated by randomly sampling from the previous one, and each layer serves as a set of centroids for the preceding layer. In contrast, Ball Trees are typically binary trees built in a top-down manner. 

Other related topic is 
\textit{Max Inner Product Search (MIPS)}~\cite{shrivastava2014asymmetric, bachrach2014speeding} which focuses on retrieving the item with the highest inner product with a query vector. Techniques, such as asymmetric LSH~\cite{shrivastava2014asymmetric} and tree- or graph-based indices~\cite{bachrach2014speeding, auvolat2015clustering}, have been proposed to accelerate MIPS. In parallel, \textit{Quantile Sketches}~\cite{greenwald2001space, karnin2016optimal} provide compact data summaries that approximate rank-based statistics, enabling efficient quantile estimation in streaming or large-scale settings. Unlike these methods, our work on \textit{Random-Access Retrieval} aims to directly access the $i$-th element according to an arbitrary scoring function.

\stitle{Nearest Neighbor Search}
A closely related problem to random-access similarity search (Problem~\ref{prob:similarity}) is \emph{Approximate Nearest Neighbor} (ANN) search, which aims to retrieve the top-$k$ most similar points to a query $q$. ANN has been extensively studied in the context of vector databases~\cite{ann1,ann2}, and existing methods broadly fall into three categories: (a) graph-based indexing and traversal methods, (b) quantization-based approaches that compress the search space, and (c) hashing techniques inspired by locality-sensitive hashing.

Graph-based methods construct a proximity graph over the dataset and perform greedy graph traversal at query time~\cite{rngs, preparata2012computational}. Prominent examples include Navigable Small-World (NSW) graphs and their hierarchical variant HNSW~\cite{nsw,malkov2018efficient}, as well as Navigating Spread-out Graphs (NSG)~\cite{nsg} and DiskANN~\cite{diskann}.

Quantization-based techniques aim to accelerate kNN search by compressing high-dimensional vectors. Widely used approaches include Product Quantization (PQ)~\cite{pq} and Inverted File (IVF) indexing~\cite{douze2024faiss,johnson2019billion}. PQ and its variants decompose vectors into multiple low-dimensional subspaces and quantize each independently, enabling fast approximate distance computation via lookup tables. For large-scale datasets, PQ is commonly combined with coarse partitioning schemes such as IVF, which clusters the dataset into centroid-based cells approximating Voronoi regions.

The third category consists of hash-based methods such as locality-sensitive hashing, which provide theoretical guarantees~\cite{guarantee, eucli, mulprobe} but are known to degrade in high-dimensional real-world settings~\cite{jafari2021survey}.

While these methods are designed to solve similarity search under distance metrics such as $\ell_2$ or cosine similarity, they primarily target regimes where $k$ is a small constant. As demonstrated in our experiments, their performance degrades substantially when applied to random-access retrieval at an arbitrary rank~$i$. In contrast, our algorithm is explicitly designed to address this setting.

\color{black}
\section{Extension to Non-linear Scoring Functions}~\label{app:nonlinear}
In Section~\ref{sec:def}, we defined ranked retrieval with respect to \emph{linear scoring functions}, where the score of a point $\point \in \Reals^d$ is computed as a weighted sum of its attributes, i.e.,
\[
\score_{\ef}(\point) = \ef^\top \point = \sum_{j=1}^{d} \ef[j]\point[j],
\]
for a weight vector $\ef \in \Reals^d$.

While our formal development focuses on linear scoring functions, the scope of our framework naturally extends to a broad class of non-linear scoring functions that can be expressed as a {\em linear combination of derived terms}. Such representations are standard in machine learning and data analysis, including linear regression models with feature transformations~\cite{montgomery2021introduction} and monotonic transformations commonly used in ranking systems~\cite{asudeh2018obtaining}.

The key idea is to introduce non-linear terms as additional \emph{derived attributes}, thereby lifting the original feature space into a higher-dimensional space in which the scoring function becomes linear. All results developed for linear scoring functions are then applied directly in the expanded space.

\paragraph{Illustrative Example.}
Consider the following non-linear scoring function over two attributes $\{x, y\}$:
\[
\score(\point) = 2x^2 + y^2 - 3x/y + 5xy - x + 2y.
\]
This function can be equivalently expressed as a linear scoring function over the expanded attribute set
\[
\{x_1=x^2, x_2=y^2, x_3=x/y, x_4=xy, x_5x, x_6=y\}.
\]
Then, using the new attributes, the function can be written as 
\[
\score(\point) = 2x_1 + x_2 - 3x_3 + 5x_4 - x_5 + 2x_6.
\]

\paragraph{Real-World Example.}
As a concrete real-world instance, consider the ranking function used by CSMetrics~\cite{csmetrics} to rank computer science research institutions. The ranking combines two attributes for each institution: the number of \emph{measured} citations ($M$) and an estimate of \emph{predicted} future citations ($P$). Given a parameter $\alpha \in [0,1]$, the scoring function is defined as
\[
\score_{\alpha}(\point) = (M)^{\alpha}(P)^{1-\alpha}.
\]

Although this function is multiplicative and non-linear in the original attributes, it can be transformed into a linear form by introducing the derived attributes~\cite{asudeh2018obtaining}
\[
x_1 = \log(M), \qquad x_2 = \log(P),
\]
Using these attributes, the scoring function can be rewritten as
\[
\score^{\ell}_{\alpha}(\point) = \alpha x_1 + (1-\alpha)x_2.
\]

Since the logarithm is a monotonic transformation, the original function $\score_{\alpha}(\cdot)$ and its linearized form $\score^{\ell}_{\alpha}(\cdot)$ induce identical rankings. Consequently, all algorithms and guarantees developed in this paper apply unchanged to such non-linear scoring functions after feature transformation.

\color{black}
\section{Missing Details for Reducing RAS to RAR}\label{app:reduction}

\subsection{Euclidean Distance: Proof of Lemma~\ref{lem:reduction-1}}\label{app:reduction-1}

\begin{proof}
Consider the Euclidean distance function 
\[
\dist(p,q)=\|p - q\|_2 =\sqrt{\sum_{j=1}^d(p[j]-q[j])^2}.
\]
We first note that, instead of using $\dist(p,q)$, one could use $\dist(p,q)^2$ as the function without loss of generality.

The expansion of the term $\dist(p,q)^2$ provides the basis of a transformation that connects \similarity\ and \css\ problems. This expansion is  given by,
\begin{equation}\label{eqn:l2-norm-nn}
    \|p-q\|_2^2=\|p\|^2_2-2q^\top p+\|q\|^2_2
\end{equation}
The term $q^T\cdot p$ is akin to computing the score based on the dot product (\emph{i.e.}, a linear function), where $q$ is the weight vector.

\vspace{2mm}
The $d$ dimensional \similarity\ problem can be reduced to the $d+1$ dimensional \ith problem using a transformation. The transformation relies on the notion of Paraboloid lifting~\cite{har2011geometric} given by $\phi:\mathbb{R}^d\rightarrow\mathbb{R}^{d+1}$. The function $\phi$ transforms a $d$ dimensional point to a $d+1$ dimensional point in a new transformed space. Following literature, the paraboloid lifting is defined as below,
\begin{equation}\label{eqn:paraboloid-transform}
    \phi(p)=(p[1], p[2], \ldots, p[d], \|p\|^2_2)
\end{equation}

As a first step, during the pre-processing stage, the paraboloid lifting function $\phi$ is applied to every point in the dataset $\mathcal{D}$ to obtain $d+1$ dimensional dataset. When an \similarity\ query $q$ arrives, we transform the query $q\in \mathbb{R}^d$ to a $d+1$ dimensional \ith\ scoring function $f=(-2q[1], -2q[2],\ldots, -2q[d],1)$. Applying the scoring function $f$ to any point $p$ in the dataset, we obtain the term 
\[\score_f(p) = f^\top p =  \|p\|^2_2-2q^\top p = \|p-q\|_2^2-\|q\|_2^2\]
Note that the term $\|q\|^2_2$ is still needed to obtain the exact same value as the $\dist(p,q)^2$ term in Equation~\ref{eqn:l2-norm-nn}. But, as the term $\|q\|^2_2$ is independent of $p$, and is missing for all  $p_i\in\mathcal{D}$ as part of the computation of $f$, there is no change in ranking. That is, while the scores for all points in $\data$ are offset by $\|q\|_2^2$, the ranking between points does not change. 
Hence, the transformation suffices to answer any $\ell_2$ norm based \similarity\ query using a transformed \ith\ query with rank set to $1$. A visual representation of this lifting transformation is shown in Figure~\ref{fig:lifting}.
\end{proof}

\subsection{Cosine Distance: Proof of Lemma~\ref{lem:reduction-2}}\label{app:reduction-2}
\begin{proof}
For a query $q$, the cosine distance is given by,
\begin{equation}\label{eq:cosine}
    \dist(p,q)=1-\frac{q^\top p}{\|q\|_2\cdot\|p\|_2}.
\end{equation}

The numerator of the Equation~\ref{eq:cosine}, $q^\top p$, can be interpreted as a dot product between the linear function $q$ and $p$. Therefore, 
\begin{align}
    \nonumber \dist(p,q)
    &=1-\left(\frac{q}{\|q\|_2}\right)^\top\frac{p}{\|p\|_2}\\
    \nonumber
    &= 1 + \sum_{j=1}^d \frac{-q[j]}{\|q\|_2} \frac{p}{\|p\|_2}
\end{align}

Removing the constant term from the scoring function, the ranking generated by Equation~\ref{eq:cosine2} is the same as Equation~\ref{eq:cosine}.

\begin{equation}\label{eq:cosine2}
    \score(p) = \sum_{j=1}^d \frac{-q[j]}{\|q\|_2} \frac{p}{\|p\|_2}
\end{equation}

Therefore, transforming each data point $p$ to \(\frac{p}{\|p\|_2}\) in the transformed space,
given a query $q\in\Reals^d$, the corresponding weight vector of the scoring function is \(f=(\frac{q[1]}{\|q\|_2},\ldots, \frac{q[d]}{\|q\|_2}).\)


\end{proof}
\section{Our Algorithm Based on $k$-th Level of Arrangement}\label{sec:app:kthlevel}

In this section, we describe the details of \textsc{KthLevel} algorithm. We begin by introducing key concepts from computational geometry, and then describe our algorithm in detail.

\subsection{Preliminaries}\label{sec:kth:prelim}
\stitle{Duality}
We use the classic dual space transformation that maps points to hyperplanes and vice versa. This transformation is well-defined and preserves intersection relationships~\cite{edelsbrunner1987algorithms, chazelle1985power}. Formally, the dual of a point \(\point \in \mathbb{R}^d\) is defined as:
\[
    \dual(\point) := \left\{ \mathbf{x} \in \mathbb{R}^d \mid \mathbf{x}^\top \point = 1 \right\}.
\]
In other words, the transformation \(\dual(\cdot)\) maps a point in the \emph{primal space} to a hyperplane in the \emph{dual space}. This hyperplane is represented by the equation:
\[
    h: \sum_{i=1}^d p[i] \cdot x_i = 1.
\]

\begin{figure*}[t]
    \centering
    \includegraphics[width=0.8\linewidth]{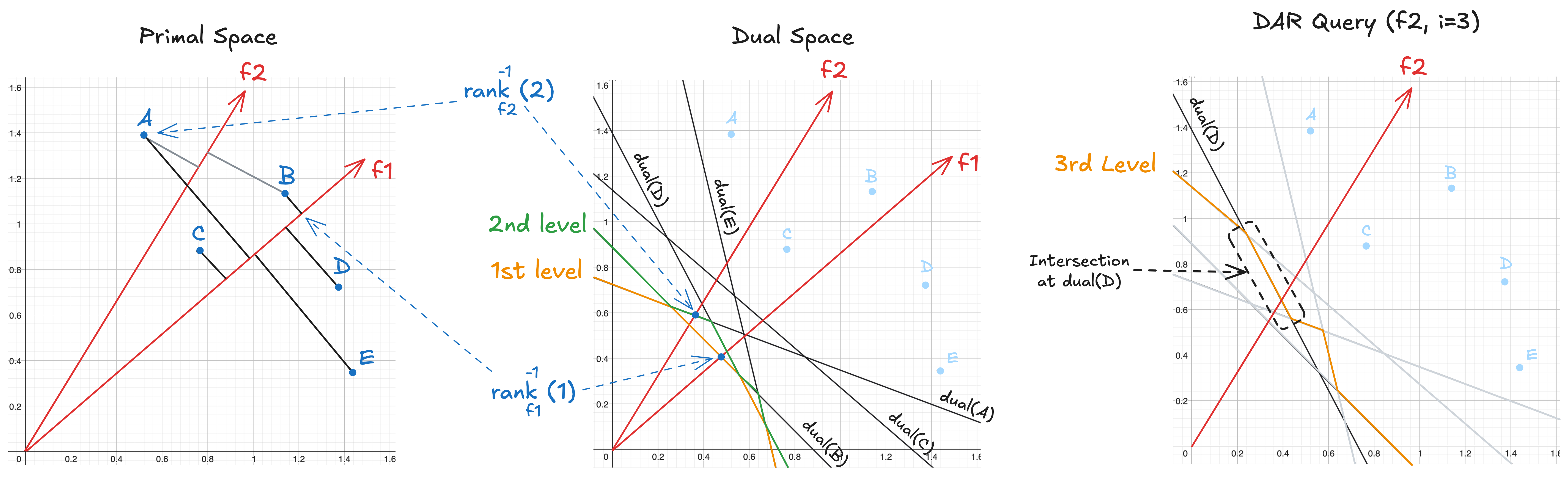}
    \vspace{-4mm}
    \caption{A 2D representation of dual space and the $k$-th level of arrangements (For simplicity, only the first quadrant is presented). The right figure shows an example of running \textbf{KthLevel} for $i = 3$. The result of this query is the point $D$.}
    \label{fig:kthlevel}
\end{figure*}

We refer to the original input domain as the \emph{primal space}, and the space of dual hyperplanes as the \emph{dual space}. Figure~\ref{fig:kthlevel} provides a visual illustration of this transformation. The left figure shows the primal space, while the middle figure shows the dual hyperplanes in the dual space.
The following key observation follows from this duality. A formal proof can be found in~\cite{edelsbrunner1987algorithms}.

\begin{lemma}\label{lm:kthlevel}
Let \(\ef\) be a scoring function (i.e., a direction vector), and let \(\point^{(1)}, \ldots, \point^{(n)}\) be the points sorted in descending order according to their score under \(\ef\), so that \(\ef^\top \point^{(1)} > \ef^\top \point^{(2)} > \cdots\). Then, for any \(a < b\), the intersection point of \(\dual(\point^{(b)})\) with the ray in direction \(\ef\) lies closer to the origin than that of \(\dual(\point^{(a)})\). \hfill\qed
\end{lemma}
For instance, in Figure~\ref{fig:kthlevel}, the point \(B\) has the highest score under \(\ef_1\), and thus its dual \(\dual(B)\) intersects the ray \(\ef_1\) in the dual space earlier (closer to origin) than any other point's dual.
\subsubsection{$k$-th Level of Arrangement}
Consider the set of dual hyperplanes \(\mathcal{H} = \{\dual(\point) \mid \point \in \data\}\).
$\mathcal{H}$ defines a dissection of $\Reals^d$ into
connected convex cells, called the {\em arrangement} of the hyperplanes~\cite{edelsbrunner1987algorithms}.

To simplify the explanations, let $d=2$.
Sweep a ray \(\ef\) counterclockwise from the \(x\)-axis to the \(y\)-axis. At each orientation of \(\ef\), the ray intersects all dual hyperplanes (dual lines) in \(\mathcal{H}\). These intersection points can be sorted by their distance from the origin. The \emph{$k$-th level} is the locus of points that lie on the $k$-th closest dual line intersected by \(\ef\), as the direction of \(\ef\) varies continuously~\cite{edelsbrunner1987algorithms}. 

In general \(d\)-dimensional space, consider a ray \(\ef\) originating at the origin and sweeping over all directions on the unit sphere \(\mathbb{S}^{d-1}\). For each direction \(\ef\), the ray intersects the hyperplanes in \(\mathcal{H}\) in up to \(n\) distinct points, which can again be ordered by increasing distance from the origin. The \emph{$k$-th level} is the set of all such $k$-th intersection points across all directions \(\ef \in \mathbb{S}^{d-1}\).

Figure~\ref{fig:kthlevel} illustrates the first and second levels of the arrangement in \(\mathbb{R}^2\), shown in orange and green, respectively.

\subsection{\textsc{KthLevel} Algorithm}
In this subsection, we formally present our algorithm, \textsc{KthLevel}. We begin by describing the \emph{preprocessing} phase, followed by a discussion of the \emph{query phase}, and conclude with a theoretical analysis of the algorithm's performance. The overall structure of the algorithm applies uniformly across dimensions \(d \geq 2\); however, for the purpose of analysis, we distinguish between the planar case (\(d = 2\)) and the general higher-dimensional case (\(d > 2\)).

\stitle{Preprocessing}
During the preprocessing phase, we construct the
arrangement of the dual hyperplanes induced by the point set \(\data\), and the levels of the arrangement for all \(1 \leq k \leq n\). We assume access to an oracle \(\mathcal{K}\) that, for a given \(i\), returns an efficient data structure representing the $i$-th level of the arrangement~\cite{edelsbrunner1987algorithms}. 
By calling this oracle for each \(i \in [n]\), we obtain the complete set of structures needed for query processing. 
Each such structure allows constant-time access to the hyperplanes intersecting on each face of the arrangement, as well as the points in the primal space corresponding to those hyperplanes~\cite{edelsbrunner1987algorithms}.

\stitle{Query Phase}
In the query phase, we are given a scoring function \(\ef\) and a rank \(i\), and the goal is to retrieve the point \(\point^{(i)}\), which is the \(i\)-th highest-ranked point with respect to \(\ef\).
By Lemma~\ref{lm:kthlevel}, the 
intersection of the \(i\)-th level of arrangement of dual hyper-planes with \(\ef\) corresponds with the desired point \(\point^{(i)}\).
Algorithm~\ref{alg:kthlevel} shows the pseudo-code of our process for finding this intersection.

Consider the 3rd level of the arrangement highlighted (in orange) in the Figure~\ref{fig:kthlevel} (right figure).
In 2D, the $k$-th level is a non-convex chain of line segments, each belonging to a dual line.
Now, let the query function be $f_2$, while $i=3$. Our objective is to find $p^{(3)} = \rank^{-1}_{\ef_2}(3)$.
The line segment in the 3rd level that intersects with the vector of $f_2$ is $\mathsf{dual}(D)$. Hence, the point $p^{(3)} =D$ in this example.
Note that applying a binary search on the endpoints of the line segments of the 3rd level, one can find the intersection of any function $f$ with any level $i$ in a time logarithmic to the number of line segments.

From the above example, observe that the $k$-th level is a chain of line segments. Similarly, for a general $d\geq 2$, the $k$-th level is a chain of \((d-1)\)-faces.
A \(d'\)-face of an arrangement is a \(d'\)-dimensional convex cell, formed by the intersection of \(d - d'\) hyperplanes. For example, 0-faces are vertices (points), 1-faces are edges (segments or lines), and so on.

Each function $f$, in such cases, is an origin-anchored vector that intersects one of the faces of each level $i, \forall i\in [n]$.
Projecting the endpoints of these faces on the surface of the unit sphere \(\mathbb{S}^{d-1}\) forms an arrangement in the $(d-1)$-dimensional space. For example, when $d=3$, each $k$-th level is a chain of plane segments, and their projection onto the unit ball (centered at the origin) is an arrangement of lines.

Now, consider the intersection point of the query function $f$ with the unit sphere \(\mathbb{S}^{d-1}\).
Finding the face segment of a level $i$ with a function $f$ is equivalent to finding the cell of the arrangement on the surface of the sphere that intersects with $f$.
This problem is known as {\em point location} on the arrangement of hyperplanes~\cite{meiser1993point}.
Line~\ref{alg1:line2} of Algorithm~\ref{alg:kthlevel} uses a state-of-art 
point-location algorithm for finding the intersection of $f$ with the $i$-th level.

\begin{algorithm}[t]
\caption{\textsc{KthLevel-Query}$(\ef, i)$}
\label{alg:kthlevel}
\begin{algorithmic}[1]
\Require The scoring function $\ef$ and the order $i$
\Ensure The $i$-th highest scored point according to $\ef$: $\point^{(i)}$

\State $Level_i \gets $ preprocessed data structure\Comment{Get the $i$-th level}

\State $face \gets \mathsf{Intersection(\ef, Level_i)}$\label{alg1:line2} \Comment{Get the intersection}

\State $h \gets $ The hyperplane corresponding to $face$ in arrangement

\State $\point^{(i)} \gets \dual^{-1}(h)$

\State \Return $\point^{(i)}$

\end{algorithmic}
\end{algorithm}

\subsection{Analysis of \textbf{KthLevel} Algorithm}
\stitle{Preprocessing}
In 2D, the complexity of the $k$-th level in an arrangement of $n$ lines in the plane is \(O(nk^{1/3})\)~\cite{dey1998improved}. 
Hence, considering an arbitrary level $i=O(n)$, the complexity of the $i$-th level is 
\(O(n^{4/3})\). 
The total preprocessing time to construct all levels of the arrangement is \(O(n^2)\), where each individual level can be built in \(O(n \log n + nk)\) time~\cite{everett1993optimal}.
In 3D, the complexity of the $k$-th level is \(O(nk^{3/2})\)~\cite{sharir2000improved}. As a result, 
the complexity of each level $i\in [n]$ is bounded by \(O(n^{5/2})\).
The complexity of the arrangement in the general $d$-dimensional setting is $O(n^d)$, and the total time required to construct it is also \(O(n^d)\)~\cite{edelsbrunner1987algorithms,agarwal1998constructing}.

\stitle{Query Time}
During the query time, Algorithm~\ref{alg:kthlevel} finds the intersection of the $i$-th level with the vector of the function $f$. Given that the total number of faces in each level $i\in [n]$ is bounded by $O(n^d)$, identifying this intersection using a state-of-the-art algorithm for point location on the arrangement of hyperplanes is in $O(\log n)$, ignoring the constant factors that depend on $d$~\cite{meiser1993point}.
\section{Analysis of the \textsc{Eps2D} Algorithm}\label{sec:app:eps2d}
We analyze the time and space complexity of the \textsc{Eps2D} algorithm (discussed in section~\ref{sec:eps2d}) step-by-step.

\emph{(Step i)}   
Finding the $i_\ell$-th and $i_u$-th point in $\epssample_\eps$ of size $m$ takes $O(m)$ using the extension of the median finding algorithm.

\emph{(Step ii)} 
This step involves solving a simplex range searching query in 2D (with stripe ranges). Using the Matoušek’s efficient range searching technique~\cite{matouvsek1991efficient}, this can be done in \( O(\sqrt{n} + |\result|) \) time with \( O(n) \) space.

\emph{(Step iii)} 
We perform a half-space range counting query. Again using Matoušek’s partition tree~\cite{matouvsek1991efficient}, this takes \( O(\sqrt{n}) \) time and \( O(n) \) space.  
(Step iv) 
Finally, we sort the set \( \result \), which requires \( O(|\result| \log |\result|) \) time.

\stitle{Space Complexity}  
All operations use data structures with linear space requirements. Thus, the total space usage is linear in input.

\stitle{Query Time Complexity}  
To analyze the runtime precisely, we bound the size of the result set \( \result \) as follows:

\begin{lemma}\label{lm:output-size}
    The size of the result set \( \result \) is \( O(\eps n) \).
    \hfill\qed
\end{lemma}

\begin{proof}
    The result set \( \result \) consists of the points in the dataset \( \data \) that lie between the \( i_\ell \)-th and \( i_u \)-th ranked points in the sample \( \epssample_\eps \). Define the corresponding score interval as the stripe range:
    \[
        S = \{x \in \Reals^d \mid \ef^\top q_\ell \leq \ef^\top x \leq \ef^\top q_u\}.
    \]
    Then, \( |\result| = |S \cap \data| \). Since \( \epssample_\eps \) is an \(\eps\)-sample of \( \data \) for stripe ranges, we have (with high probability):
    \(
        \left| \frac{|S \cap \epssample_\eps|}{|\epssample_\eps|} - \frac{|\result|}{|\data|} \right| \leq \eps.
    \)
    From Equation~\ref{eq:klku}, we know that \( |S \cap \epssample_\eps| = |i_u - i_\ell| \leq 2m\eps + 2 \). Hence:
    \begin{align*}
        \left| \frac{|S \cap \epssample_\eps|}{m} - \frac{|\result|}{n} \right| \leq \eps 
        \Longrightarrow \frac{|\result|}{n} \leq \frac{2m\eps + 2}{m} + \eps 
        \Longrightarrow |\result| = O(\eps n).
    \end{align*}
\end{proof}

In 2D, the size of the \(\eps\)-sample satisfies (see Section~\ref{sec:background}):
\(
    m = O\left( \frac{1}{\eps^2} \log \frac{1}{\eps} \right).
\)
Combining all steps, the total runtime is:
\[
    T(\textsc{Eps2D}) = O\left(m + \sqrt{n} + |\result| \log |\result|\right).
\]
Using \( |\result| = O(\eps n) \) and choosing \( \eps = n^{-1/3} \) (which minimizes the overall runtime), we get:
\[
    T(\textsc{Eps2D}) = O\left(n^{2/3} \log n\right).
\]
\vspace{-2mm}
\section{Missing Proofs}\label{sec:app:proofs}
\subsection{Proof of Lemma~\ref{lm:boundaries}}
\begin{proof}
    Let \(\mathcal{H}\) be the half-space defined by:
    $
        \mathcal{H} = \{x \in \Reals^d \mid \ef^\top x \geq \ef^\top \point^{(i)}\}.
    $
    This half-space is perpendicular to the direction vector \(\ef\) and contains exactly the top \(i\) points of \(\data\) ranked by \(\ef\), including \(\point^{(i)}\). In other words,
    $
        |\mathcal{H} \cap \data| = i.
    $
    
    Since \(\epssample_\eps\) is an \(\eps\)-sample of \(\data\) for half-space (and slab) ranges, it preserves the proportion of points in any half-space up to additive error \(\eps\). Therefore, 
    \[
        \left| \frac{|\mathcal{H} \cap \epssample_\eps|}{|\epssample_\eps|} - \frac{|\mathcal{H} \cap \data|}{|\data|} \right| \leq \eps.
    \]
    Substituting \(|\mathcal{H} \cap \data| = i\), and denoting \(n = |\data|\) and \(m = |\epssample_\eps|\), we obtain:
    \[
        \left| \frac{|\mathcal{H} \cap \epssample_\eps|}{m} - \frac{i}{n} \right| \leq \eps,
    \]
    which implies:
    \[
        m\left( \frac{i}{n} - \eps \right) \leq |\mathcal{H} \cap \epssample_\eps| \leq m\left( \frac{i}{n} + \eps \right).
    \]
    Rounding to integers, define:
    \[
        i_\ell = \left\lfloor m\left( \frac{i}{n} - \eps \right) \right\rfloor, \quad
        i_u = \left\lceil m\left( \frac{i}{n} + \eps \right) \right\rceil.
    \]
    These bounds indicate that the number of points in \(\epssample_\eps\) with score at least \(\score_\ef(\point^{(i)})\) lies between \(i_\ell\) and \(i_u\). Equivalently, the score \(\score_\ef(\point^{(i)})\) lies between the scores of the \(i_u\)-th and \(i_\ell\)-th ranked points in \(\epssample_\eps\). 
    We conclude that:
    \[
        \score_\ef(q_\ell) \leq \score_\ef(\point^{(i)}) \leq \score_\ef(q_u),
    \]
    as desired.
\end{proof}
\section{Dynamic Setting}\label{sec:app:dynamic}
In this section, we discuss the dynamic setting. We consider two types of dynamic operations: \texttt{Insert(p)} and \texttt{Delete(p)}. For the \textsc{KthLevel} algorithm, these operations are supported following the literature on dynamically maintaining levels of arrangements~\cite{chan2020dynamic, agarwal2019dynamic}. Below, we focus on the \textsc{EpsHier} algorithm; analogous results can be shown for \textsc{EpsRange}.

\stitle{\texttt{Insert(p)}}
This operation adds a new point $p$ to the dataset $\data$.  
For \textsc{EpsHier}, we must update both the $\eps$-sample $\epssample_\eps$ and the hierarchical sampling structure described in Section~\ref{sec:hierarchical}.  
Following Theorem~\ref{thm:epssample}, we want to maintain an unbiased random subset of size $|\epssample_\eps|$ of $\data$ as an $\eps$-sample.
Our goal is to maintain such a random subset without resampling the entire dataset.  
This can be efficiently achieved using \emph{reservoir sampling}~\cite{vitter1985random}, which allows us to maintain a uniform random sample as new elements arrive.  
To update the hierarchical structure used for the \srs problem, we apply this reservoir sampling procedure bottom-up, from the base layer to higher layers. At each layer, we decide whether to include the new point, as each layer is a random sample of the one below it. This process takes $O(\log n)$ time since the index contains $\log n$ layers.

\stitle{\texttt{Delete(p)}}
This operation removes $p$ from $\data$.  
Under the \emph{backing sample} method~\cite{gibbons2002fast}, if $p \in \epssample_\eps$, we delete it from the sample.  
If the sample size drops below a predefined lower bound $L$, we \emph{rebuild} $\epssample_\eps$ by resampling all the points. The amortized time for this operation (based on the definition of $L$) is $O(1)$~\cite{gibbons2002fast}.
For the hierarchical structure, each layer $\ell$ is maintained similarly, with its own bound $L_\ell$, and each layer $\ell$ is the sample maintained at level $\ell\!-\!1$.  
Hence, a rebuild at level $\ell$ only scans the immediately lower level, not the entire dataset, keeping the amortized cost low. In other words, this rescanning at layer $\ell$ takes $O(|\mathcal{L}_{\ell - 1}|) = O(\frac{n}{2^{\ell-1}})$, but the amortized is still $O(1)$ for this layer and it takes $O(\log n)$ to resample all the top layers. As a result,
the total amortized time for this operation is $O(\log n)$ for the hierarchical structure.

\section{Analysis of the \textsc{EpsRange} Algorithm}\label{sec:app:epsrange:analysis}
The analysis is similar to \textsc{Eps2D}, however, the time complexity of the algorithm depends on the specific range counting and simplex range searching we use at steps ii and iii.

Step i takes $O(m)$ for finding the thresholds in the $\eps$-sample. Step ii requires solving the \srs~ problem. Following the simplex range searching algorithms, we can solve this problem in $O(n^{1 - 1/d} + |\result|)$ with $O(n)$ space~\cite{matousek1992reporting}.
To report the exact response to \ith, i.e., $\point^{(i)}$, we proceed to steps iii and iv. In step iii, we solve an instance of range counting in $d$ dimensions, which takes $O(n^{1 - 1/d})$~\cite{matousek1992reporting} time and linear space. Finally, in step iv, we sort $\result$, which takes $O(|\result| \log |\result|)$.

\stitle{Space Complexity} The space complexity of this algorithm is linear to the input size, since all structures require linear space.

\stitle{Time Complexity for the $\kappa$-\ith\ problem}
The total query time of the algorithm for finding the conformal set is:
\[
    T(\textsc{EpsRange}_{\textit{$\kappa$-\ith}}) = O(m + n^{1 - 1/d})
\]
The size of $\eps$-sample $m$ is $O(\frac{d}{\eps^2}\log \frac{d}{\eps})$, since the VC-dimension of stripe range is $\delta=d+1$~\cite{har2011geometric, haussler1986epsilon}. This gives us the following query time: 
\[
    T(\textsc{EpsRange}_{\textit{$\kappa$-\ith}}) = O(\max\{n^{1 - 1/d}, \frac{d}{\eps^2}\log \frac{d}{\eps}\}) = \Omega(n^{1 - 1/d}).
\]
The output is a set $\result$ with size $O(\eps n)$, guaranteed to contain $\point^{(i)}$ (see Lemma~\ref{lm:output-size}). Note that the final time also depends on the choice of parameter $\eps$, which determines the size of output $\result$.

\stitle{Time Complexity for the \ith problem}
The exact version also includes the steps iii and iv:
\[
    T(\textsc{EpsRange}_{Exact}) = O(m + n^{1 - 1/d} + |\result| \log |\result|).
\]
Based on Lemma~\ref{lm:output-size}, we have $|\result| = O(\eps n)$. As a result, by choosing $\eps = n^{-1/3}$, we get the following runtime:
\[
    T(\textsc{EpsRange}_{Exact}) = O(\max\{n^{1-1/d}, dn^{2/3}\log n\})
\]
\begin{theorem}
    The \textsc{EpsRange} algorithm solves the relaxed version of \ith~problem with linear space usage and $O(\max\{n^{1-1/d}, \frac{d}{\eps^2}\log \frac{d}{\eps}\})$ time. The exact version can also be solved with the same space usage and $O(\max\{n^{1 - 1/d}, dn^{2/3}\log n\})$ time.
\end{theorem}
\vspace{-3mm}
\section{Analysis of Hierarchical Sampling}\label{sec:app:hierarchy}
The total space usage for storing the hierarchical layers is linear to the input, since the size of each layer decreases exponentially:
\begin{small}
\[
    \sum_{\ell \leq L} |\mathcal{L}_\ell| = \sum_{\ell \leq L} \frac{n}{r^\ell} = O(n),
\]
\end{small}
assuming $r = O(1)$.
The runtime of Algorithm~\ref{alg:hierarchical-preprocess} is dominated by the \textbf{for} loop at line 6, where the layers are constructed, and edges and area sets are updated. At each layer $\ell$, we sample $|\mathcal{L}_\ell|$ points. Then, for each point in $\mathcal{L}_{\ell - 1}$, we compute its nearest neighbor in $\mathcal{L}_\ell$. The time complexity at the layer $\ell$ is:
\begin{small}\[
    T_\ell = O\left(d\cdot |\mathcal{L}_\ell| \cdot |\mathcal{L}_{\ell - 1}|\right).
\]
\end{small}
Summing over all layers, the total runtime becomes:
\begin{small}
\[
    \sum_{\ell} T_\ell = \sum_{\ell} d\cdot \frac{n}{r^\ell} \cdot \frac{n}{r^{\ell - 1}} 
    = \sum_{\ell} \frac{dn^2}{r^{2\ell - 1}} = O(dn^2),
\]
\end{small}
again assuming $r = O(1)$.
For computing the smallest enclosing ball of each area, we use Welzl's algorithm, which runs in linear time with respect to the number of points~\cite{welzl2005smallest}.

\begin{theorem}
    The preprocessing phase of the hierarchical sampling algorithm takes $O(dn^2)$ time and uses linear space.
\end{theorem}
\section{Additional Pseudo-codes}\label{sec:app:pseudo}
\begin{algorithm}[H]
\caption{\textsc{Eps2D-Query}$(\epssample_\eps, \data, \ef, i)$}
\label{alg:eps2d-query}
\begin{algorithmic}[1]
\Require Precomputed $\eps$-sample $\epssample_\eps$ of size $m$, original dataset $\data$, scoring vector $\ef$, target rank $i$
\Ensure The exact $k$-th ranked point $\point^{(i)}$ in $\data$ under $\score_\ef$

\State Calculate $i_\ell$ and $i_u$\Comment{equation~\ref{eq:klku}}


\State $q_\ell \gets \rank^{-1}_{\epssample_\eps, \ef}(i_\ell)$,\quad
$q_u \gets \rank^{-1}_{\epssample_\eps, \ef}(i_u)$\Comment{Median of medians algorithm}

\State $\ell \gets \score_\ef(q_\ell)$,\quad
$u \gets \score_\ef(q_u)$\Comment{end of step i}


\State $\result \gets \Stripe_{\ef, \ell, u} \cap \data$ \Comment{via \textsc{SRS} query; step ii}

\State $|H_u| \gets \left| \{ \point \in \data \mid \ef^\top \point \geq u \} \right|$ \Comment{via range counting; step iii}

\State Sort $\result$ ascending by $\score_\ef$

\State \Return the $(i - |H_u|)$-th point in the sorted list\Comment{step iv}

\end{algorithmic}
\end{algorithm}

\begin{algorithm}[H]
\caption{\textsc{Hierarchical-Sampling-Preprocess}$(\data, r)$}
\label{alg:hierarchical-preprocess}
\begin{algorithmic}[1]
\Require The dataset $\data$ and the exponential decay rate $r$
\Ensure The hierarchical graph $\mathcal{G}(\data)$ which contains a list of layers $\mathcal{L}$, neighboring relations $N$, and enclosing balls of nodes $\mathcal{B}$.

\State $L \gets \lfloor\log_r n\rfloor$\Comment{Number of layers}

\State $\mathcal{L}_\ell \gets [~],\quad\forall \ell \leq L$\Comment{Placeholder for the layers}

\State $\mathcal{L}_0 \gets \data$\Comment{Base layer}

\State $\mathcal{A}_0(p) = \{\point\},\quad\forall p\in \mathcal{L}_0$\Comment{Area of nodes}

\State $N_0(p) \gets \varnothing,\quad\forall p\in\mathcal{L}_0$\Comment{Neighbors}

\For{layer $\ell = 1$ to $L$}
    \State $\mathcal{L}_{\ell} \gets \textit{Random sample of size $\frac{|\mathcal{L}_{\ell - 1}|}{r}$ from $\mathcal{L}_{\ell-1}$}$

    \State $\mathcal{A}_\ell(p)\gets \varnothing,\quad\forall p\in\mathcal{L}_\ell$
    \State $N_\ell(p) \gets \varnothing,\quad\forall p\in\mathcal{L}_\ell$

    \For{point $\point \in \mathcal{L}_{\ell-1}$}
        \State $c^* \gets \argmin_{c \in \mathcal{L}_{\ell}} Dist(p, c)$\Comment{Equation~\ref{eq:edges}}

        \State $N_\ell(c^*) \gets N_\ell(c^*) \cup \{p\}$\Comment{Add neighbor}

        \State $\mathcal{A}_{\ell}(c^*) \gets \mathcal{A}_{\ell}(c^*) \cup \mathcal{A}_{\ell - 1}(\point)$\Comment{Update area}
    \EndFor
    
\EndFor

\State $\Ball_\ell(p) \gets \mathsf{EnclosedBall}(\mathcal{A}_\ell(p)),\quad\forall \ell \leq L, p \in \mathcal{L}_\ell$

\State $\mathcal{L} \gets \{\mathcal{L}_\ell \mid \ell \leq L\}$
\State $N \gets \{N_\ell(p) \mid \forall \ell \leq L, p\in \mathcal{L}_\ell\}$
\State $\mathcal{B} \gets \{\Ball_\ell(p) \mid \forall \ell \leq L, p \in \mathcal{L}_\ell\}$

\State \Return $(\mathcal{L}, N, \mathcal{B})$\Comment{layers, edges, and enclosing balls}

\end{algorithmic}
\end{algorithm}

\begin{algorithm}[H]
\caption{\textsc{Hierarchical-Sampling-Query}$(\mathcal{G}(\data), \Stripe_{\ef, \ell, u})$}
\label{alg:hierarchical-query}
\begin{algorithmic}[1]
\Require The preprocessed graph $\mathcal{G}(\data)$ and the query stripe range $\Stripe_{\ef, \ell, u}$.
\Ensure The set of points $\mathcal{D}_o = \Stripe_{\ef, \ell, u} \cap \data$.

\State $\mathcal{D}_o \gets \varnothing$\Comment{Placeholder for result}

\State $candidates \gets \mathcal{L}_L$

\For{$curr\_layer$ from $L$ to $1$}\Comment{top-down}
    \State $tmp \gets \varnothing$\Comment{Placeholder for new candidates}
    
    \For{each point $\point \in Candidates$}
        \If{$\ball_{curr\_layer}(\point) \cap \Stripe_{\ef, \ell, u} \neq \varnothing$}
            \State $tmp \gets tmp \cup N_{curr\_layer}(\point)$\Comment{Explore the area}
        \EndIf
    \EndFor
    
    \State $candidates \gets tmp$
\EndFor

\For{$\point \in candidates$}
    \If{$\point \in \Stripe_{\ef, \ell, u}$}
        \State $\mathcal{D}_o \gets \mathcal{D}_o \cup \{\point\}$
    \EndIf
\EndFor

\State \Return $\mathcal{D}_o$

\end{algorithmic}
\end{algorithm}

\clearpage
\section{Lower Bounds and Optimality of Our Algorithms}\label{sec:app:opt}

In this section, we analyze the complexity of the \ith~problem and show the near-optimality of our proposed algorithms. Specifically, the well-studied \emph{half-space range counting} problem~\cite{chazelle1989lower, agarwal2017range} reduces to our problem. Hence, the complexity of the state-of-the-art algorithms for half-space range counting provides a lower bound for the \ith~problem.


The half-space range counting problem is defined as follows: given a point set in \(\Reals^d\) and a query half-space \(H\), determine the number of points lying within \(H\).

\begin{lemma}\label{lm:reduction}
Let \(\mathcal{A}\) be an algorithm that solves the \ith~problem with query time \(T(\mathcal{A})\) and space complexity \(S(\mathcal{A})\). Then, using $\mathcal{A}$, one can solve an instance of the half-space range counting problem using space \(S(\mathcal{A})\) and query time $\Theta(\log n \cdot T(\mathcal{A}))$.
\end{lemma}

\begin{proof}
To count the number of points within a given half-space \(H\), we can perform a binary search over \(k \in \{1, 2, \dots, n\}\). In each iteration, we query \(\mathcal{A}\) to retrieve the \(k\)-th ranked point \(\point^{(k)}\). If \(\point^{(k)} \in H\), then the count of points in \(H\) is at least \(k\); otherwise, it is less. This process requires \(\Theta(\log n)\) queries to \(\mathcal{A}\), resulting in total query time \(\Theta(\log n \cdot T(\mathcal{A}))\).
\end{proof}

It is known that the query time for half-space range counting in \(\Reals^d\) with \(m\) units of space satisfies the following lower bound~\cite{chazelle1989lower}:
\begin{align}\label{eq:lowerbound}
    \Omega\left(\frac{n}{m^{1/d} \log n}\right).
\end{align}

Applying Lemma~\ref{lm:reduction}, this implies a lower bound for the \ith~problem:
\[
T(\mathcal{A}) \cdot \log n \geq \frac{n}{m^{1/d} \log n}
\quad \Longrightarrow \quad
T(\mathcal{A}) \geq \frac{n}{m^{1/d} \log^2 n}.
\]

Now, these are the two important observations:

\paragraph{Observation 1.} If the algorithm uses linear space, i.e., \(m = O(n)\), then the query time cannot be faster than \(O(n^{1 - \frac{1}{d}})\) up to logarithmic factors.

\paragraph{Observation 2.} If the algorithm uses exponential space, i.e., \(m = 2^{\Theta(d)}\), then the lower bound becomes \(\Omega(\log n)\). In fact, our exact algorithm \textsc{KthLevel} has logarithmic query time but requires exponential space, which matches this bound. Note that querying \textsc{KthLevel} \(n\) times is equivalent to sorting all \(n\) elements, which takes \(\Omega(n \log n)\) time.

\textbf{Conclusion.} These results demonstrate that our algorithms for the \ith~problem are optimal up to logarithmic factors, given the space-time trade-offs imposed by known lower bounds.  

\section{Extended Experiments}\label{sec:app:exp}
\subsection{Dataset Details}\label{sec:app:exp:datasets}
We use both synthetic and real-world datasets for our evaluation. The synthetic datasets are generated using the Zipfian distribution, with variations in dimensionality, dataset size, and distribution parameters. The Zipfian distribution is particularly challenging for range searching due to its inherent skewness and is commonly observed in many real-world scenarios. It allows us to evaluate the algorithms under non-uniform workloads.

For our real-world experiments, we use three datasets: the \textsf{US Used Cars} dataset~\cite{mital2020uscars}, the \textsf{FIFA 2023} dataset~\cite{leone2022fifa23}, and the \textsf{US Flights} dataset~\cite{usdot_flight_delays}. The \textsf{US Used Cars} dataset contains approximately 3 million entries with 66 attributes describing vehicle specifications, market prices, and regional features, making it suitable for evaluating ranking and filtering performance over high-dimensional numeric attributes. The \textsf{FIFA 2023} dataset comprises around 300{,}000 records with 54 player and team attributes, including overall ratings, skills, and physical characteristics, and provides a structured benchmark for similarity and top-$k$ retrieval tasks. Finally, the \textsf{US Flights} dataset contains over 5 million flight records with 31 attributes collected from the U.S. Department of Transportation, including flight times, delays, cancellations, and carrier information, offering a real-world testbed with temporal and categorical dimensions for large-scale ranked retrieval evaluation.

\paragraph{Scoring Functions}
In our experiments, we generate a uniformly random unit vector $f$ as the ranking function on all the above datasets.

\stitle{Similarity Search Datasets}
For the similarity search task, we evaluate our methods on five standard benchmark datasets from
\href{https://ann-benchmarks.com}{ann-benchmarks.com}, covering both vision and language domains:
\begin{itemize}
    \item {\textsc{SIFT-128-Euclidean}~\cite{jegou2010product}:}
    Contains 1M vectors in a 128-dimensional space with $\ell_2$ (Euclidean) distance, derived from local image descriptors and widely used as a canonical benchmark for large-scale visual similarity search.
    
    \item {\textsc{Fashion-MNIST-784-Euclidean}~\cite{xiao2017online}:}
    Contains 60K vectors in a 784-dimensional space with $\ell_2$ distance, obtained from grayscale images of fashion items and representing a moderate-scale, high-dimensional vision dataset.
    
    \item {\textsc{GIST-960-Euclidean}~\cite{jegou2010product}:}
    Contains 1M vectors in a 960-dimensional space with $\ell_2$ distance, based on global image descriptors that capture coarse scene-level structure, making it significantly higher-dimensional than SIFT.
    
    \item {\textsc{GloVe-100-Angular} and \textsc{GloVe-25-Angular}:}
    Contain 1M word embedding vectors in 100- and 25-dimensional spaces, respectively, using cosine similarity, and serve as standard benchmarks for semantic similarity search in natural language processing.
\end{itemize}

\subsection{More on Stripe Range Retrieval}\label{sec:app:exp:srs}

\stitle{Preprocessing (Indexing) Phase}
During the preprocessing phase in all the baselines, an index structure is constructed. Figure~\ref{fig:srs:prep} (left) shows the memory usage of these indices as a function of the dataset size $n$, while the right plot illustrates the corresponding preprocessing time. As shown, the space overhead of \textsf{Hierarchical Sampling} is comparable to that of \textsf{KD-Tree} and \textsf{R-Tree}, both of which are commonly used in existing systems. Moreover, the space usage grows linearly with the dataset size $n$. The time required to build the \textsf{Hierarchical Sampling} index is also on par with that of the \textsf{Partition Tree} and \textsf{R-Tree}.

\begin{figure}[ht]
    \centering
    \includegraphics[width=0.49\linewidth]{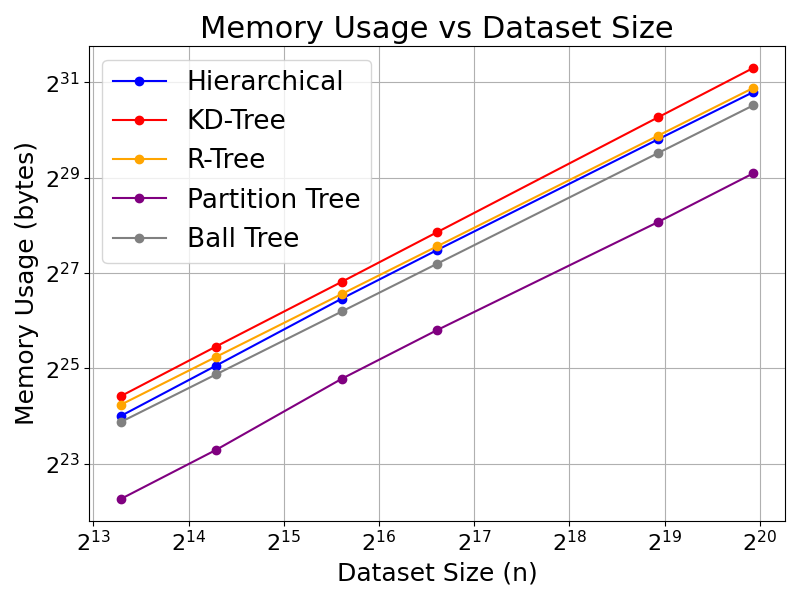}
    \includegraphics[width=0.49\linewidth]{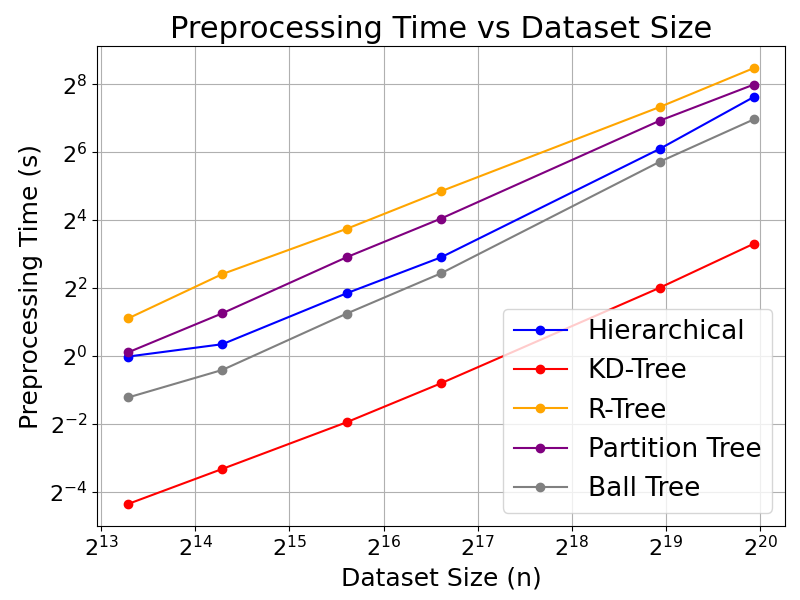}
    \vspace{-4mm}
    \caption{Comparison of \srs indexing time and size with respect to the dataset size for $d = 8$.}
    \label{fig:srs:prep}
\end{figure}

\stitle{More Comparison Reports}
Figure~\ref{fig:srs:time-width-more} shows the query time of different methods, when the value of $d$ is $4$ and $32$.

\begin{figure}[ht]
    \centering
    \includegraphics[width=0.49\linewidth]{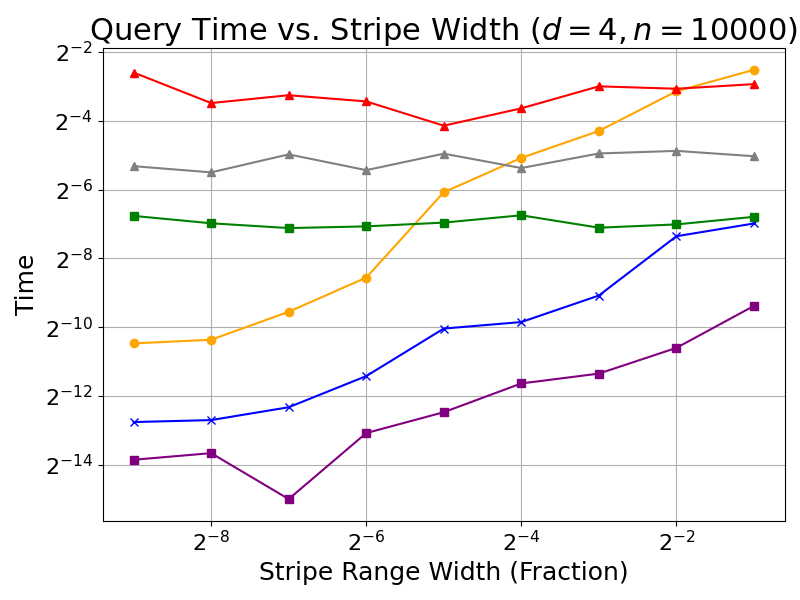}
    \includegraphics[width=0.49\linewidth]{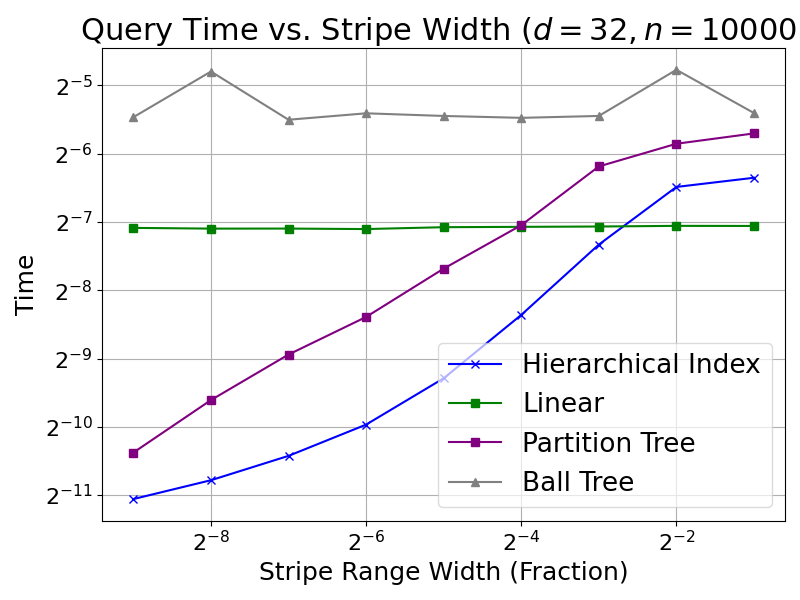}
    \caption{Comparison of \srs query time with respect to the stripe width across different dimensionalities on the synthetic data.}
    \label{fig:srs:time-width-more}
\end{figure}

Figure~\ref{fig:srs:real:more} shows the experiments for Stripe Range Searching on FIFA dataset.

\begin{figure}[H]
    \centering
    \includegraphics[width=0.6\linewidth]{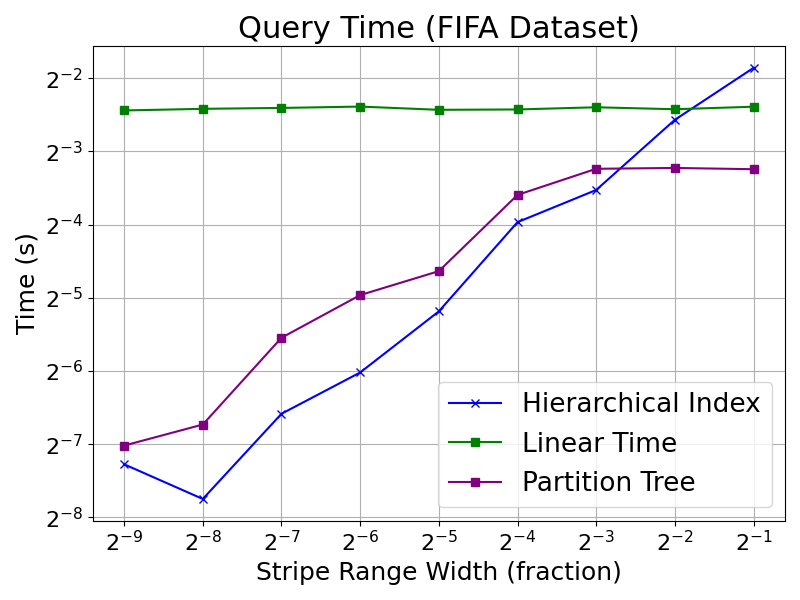}
    \caption{Comparison of the query time of Hierarchical Sampling on \srs\ problem vs the stripe range width. The scoring functions are sampled uniformly at random from a hypersphere.}
    \label{fig:srs:real:more}
\end{figure}

\subsection{More on Direct-Access and Conformal Set Retrieval}\label{sec:app:exp:rar}

\stitle{\textsf{KthLevel} Algorithm}
Figure~\ref{fig:dar:kthlevel} shows the result of running the \textsf{KthLevel} algorithm on smaller datasets in 2D. As shown in, the \textsf{KthLevel} algorithm achieves substantial speedups. However, this improvement comes at the cost of increased index size.

\begin{figure}[ht]
    \vspace{-2mm}
    \centering
    \includegraphics[width=0.49\linewidth]{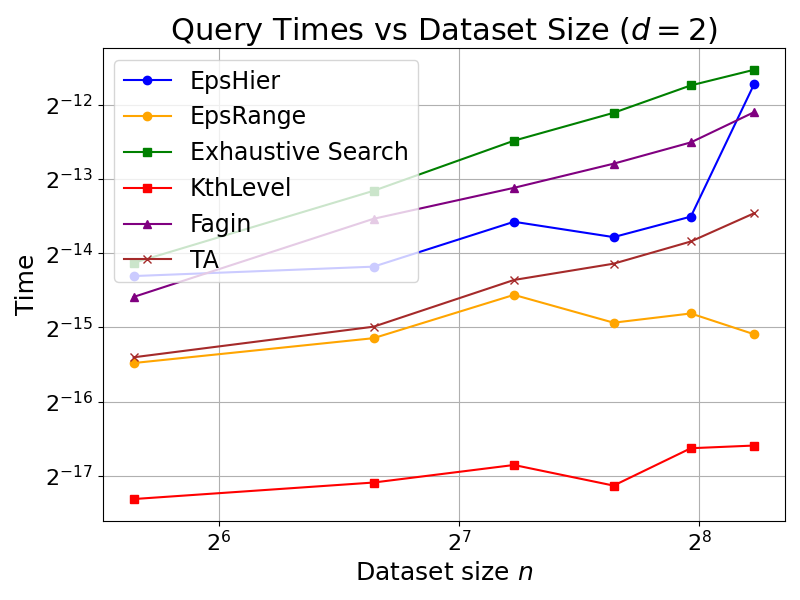}
    \includegraphics[width=0.49\linewidth]{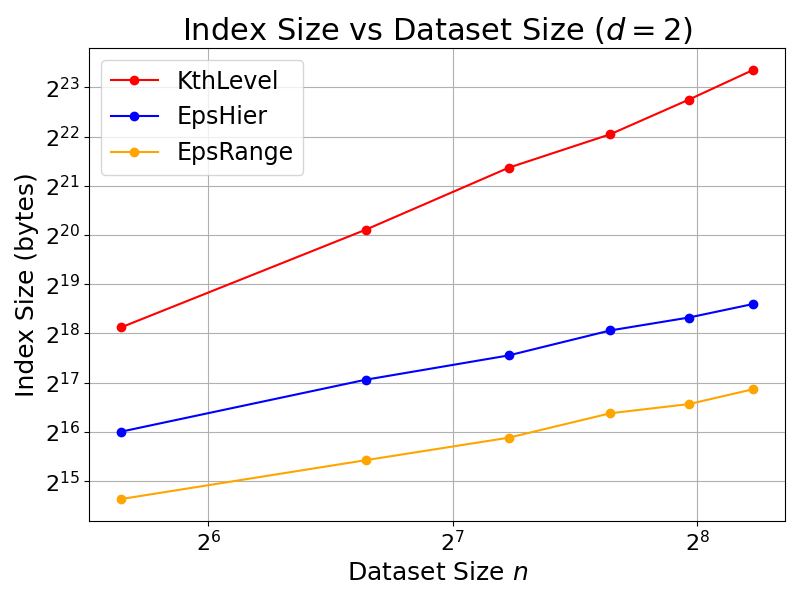}
    \vspace{-4mm}
    \caption{Result of applying \textsf{KthLevel} for solving \ith problem in 2D. Generally, \textsf{KthLevel} does not scale well with dimension $d$ and size $n$. Here, \textsf{EpsRange} in 2D is equivalent to \textsf{Eps2D} algorithm.}
    \label{fig:dar:kthlevel}\vspace{-4mm}
\end{figure}

\stitle{Other Real Datasets}
Figures~\ref{fig:dar:fifa} and~\ref{fig:dar:used-cars} compare the performance of the algorithms on the \textsf{US Used Cars} and \textsf{FIFA} datasets, where the parameters $\varepsilon$ and rank $i$ are varied across different runs.

\begin{figure}[H]
    \centering
    \includegraphics[width=0.49\linewidth]{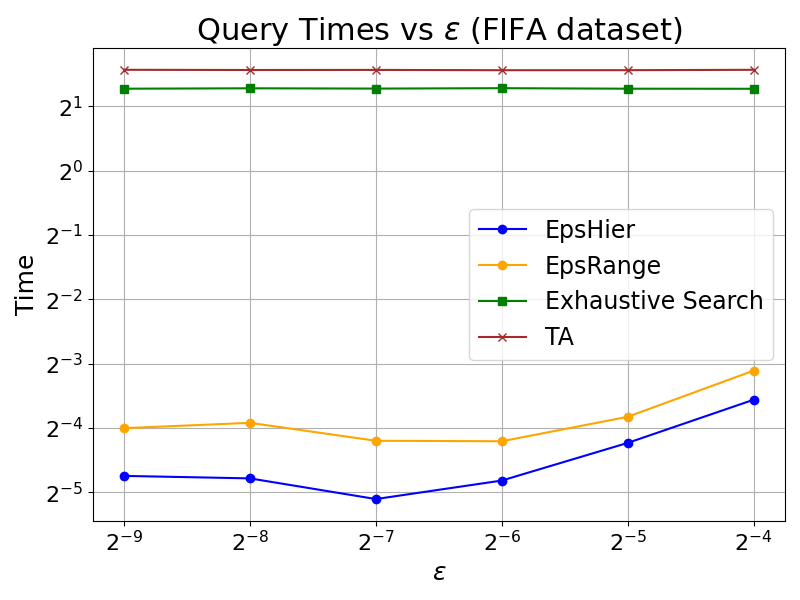}
    \includegraphics[width=0.49\linewidth]{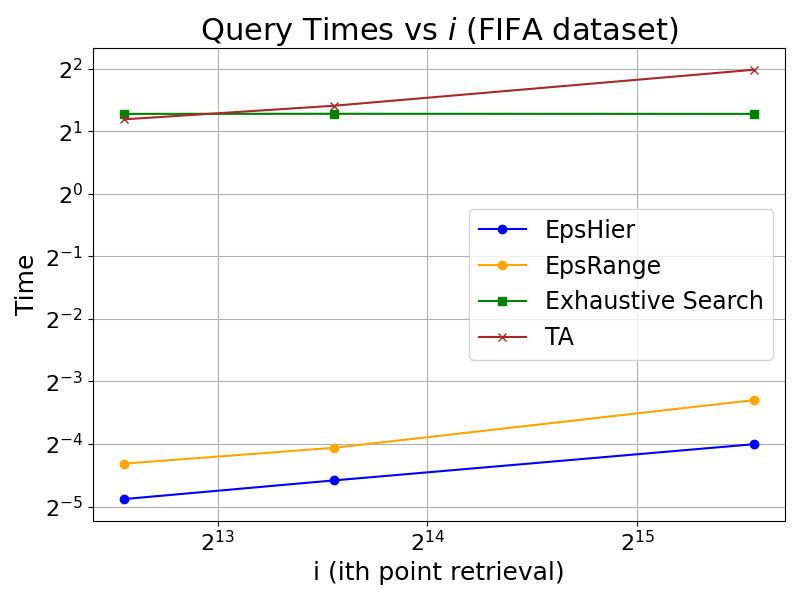}
    \caption{Comparing the query time of \textsc{EpsHier} and \textsc{EpsRange} on $\kappa$-\ith\ problem vs value of $\eps$ and rank $i$ on FIFA 2023 dataset.}
    \label{fig:dar:fifa}
    \vspace{-4mm}
\end{figure}

\begin{figure}[H]
    \centering
    \includegraphics[width=0.49\linewidth]{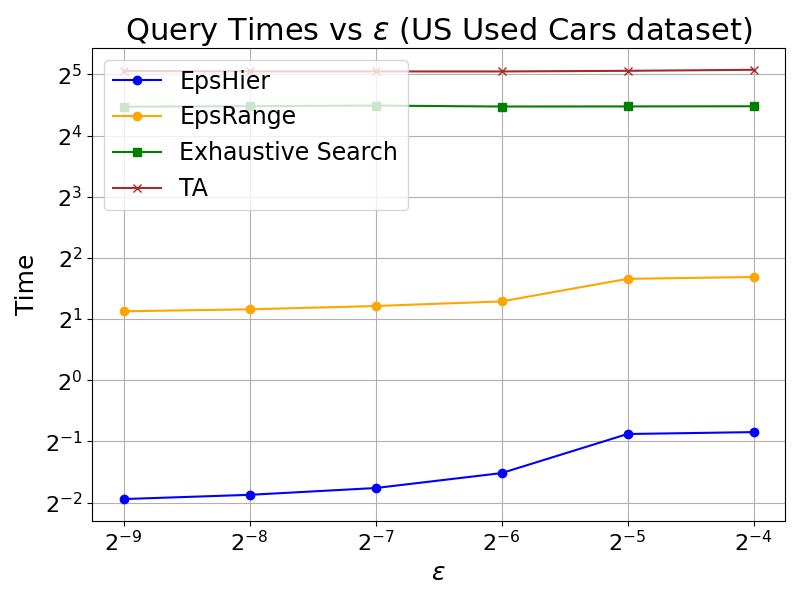}
    \includegraphics[width=0.49\linewidth]{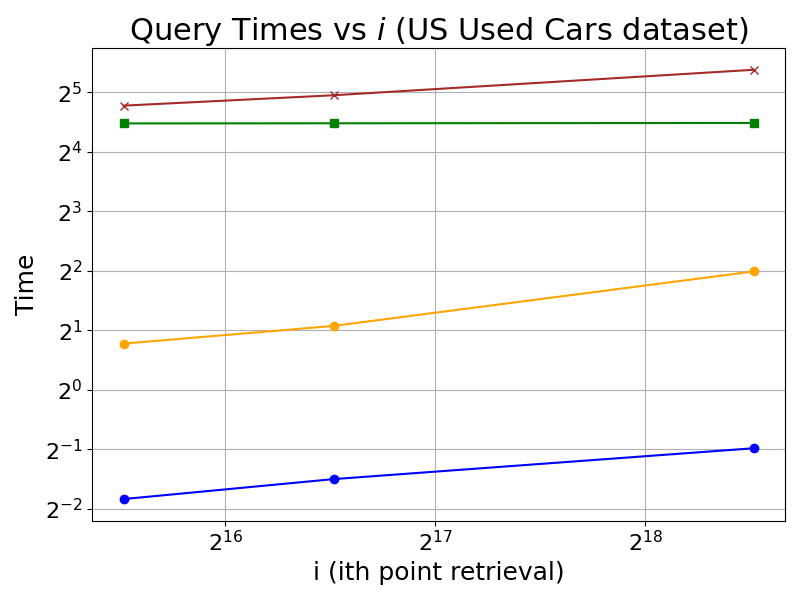}
    \caption{Comparing the query time on $\kappa$-\ith\ problem vs value of $\eps$ and rank $i$ on US Used Cars Dataset.}
    \label{fig:dar:used-cars}
    \vspace{-4mm}
\end{figure}

\subsection{Random-Access Similarity Search}\label{sec:app:exp:similarity}
\stitle{Adapting ANN Indexes to the $\kappa$-\ith\ Problem}

In this section, we describe how standard approximate nearest neighbor (ANN) indexes—such as HNSW and Product Quantization—are adapted to the random-access similarity search setting, and used to compare their performance with our proposed algorithm, \textsc{EpsHier}, in Section~\ref{sec:exp:similarity}.

Our algorithm directly solves the $\kappa$-\ith\ problem with $\kappa = \varepsilon n$ for a given $\varepsilon$. The output is a set of size $\kappa$. To evaluate recall, we check whether the ground-truth $i$-th point, denoted by $p^{(i)}$, is contained in this output set (i.e., the conformal set).

In contrast, ANN indexes are designed to return the top-$k$ nearest neighbors. To approximate the $i$-th ranked point, we query each ANN index with $k = i + 10$, where the additional constant margin ensures that the true $i$-th neighbor is covered in the best-case scenario. Depending on the index configuration, the output size may be $O(k)$ for large $k$, or bounded by a maximum size $M$. For example, in HNSW, $M$ corresponds to the beam search candidate set size. Consequently, the final output size is $\min(M, k)$.

Since this output does not directly correspond to a solution of the $\kappa$-\ith\ problem, we apply a post-processing step to ensure a fair comparison. Specifically, we select the last $\varepsilon n$ points (i.e., those with the largest distances) from the ANN output and then check whether $p^{(i)}$ lies within this subset.

\stitle{More Results on the \similarity\ Task}

Table~\ref{tab:knn-time-recall-comparison} reports the recall and query time of ANN indexes across all five benchmark datasets for varying values of $i$, ranging from small values ($i = O(1)$) to the median and third quartile of the dataset.

We observe that the performance of all ANN indexes degrades significantly as $i$ grows beyond $O(1)$. In particular, for large values of $i$, these methods fail to retrieve the required point, resulting in near-zero recall. Moreover, their query time increases substantially, as the algorithms must iterate over an increasing number of candidates and compute distances for up to $i$ points. 

In contrast, our algorithm, \textsc{EpsHier}, exhibits stable performance across all values of $i$, achieving fast query times and a guaranteed recall of $100\%$. This guarantee holds with high probability, conditioned on the sampled set forming an $\varepsilon$-sample, as discussed in Section~\ref{sec:epssampling}.

Figure~\ref{fig:ann:preprocess} compares the indexing (preprocessing) time and index size of \textsc{EpsHier} with those of the ANN indexes used in this experiment. We observe that both the index size and preprocessing time of \textsc{EpsHier} are comparable to standard ANN methods. However, unlike ANN indexes, which do not meaningfully solve the \similarity\ problem and exhibit zero recall for large $i$, our approach provides provable correctness guarantees across the entire range of $i$.

\begin{figure}[ht]
    \centering
    \includegraphics[width=\linewidth]{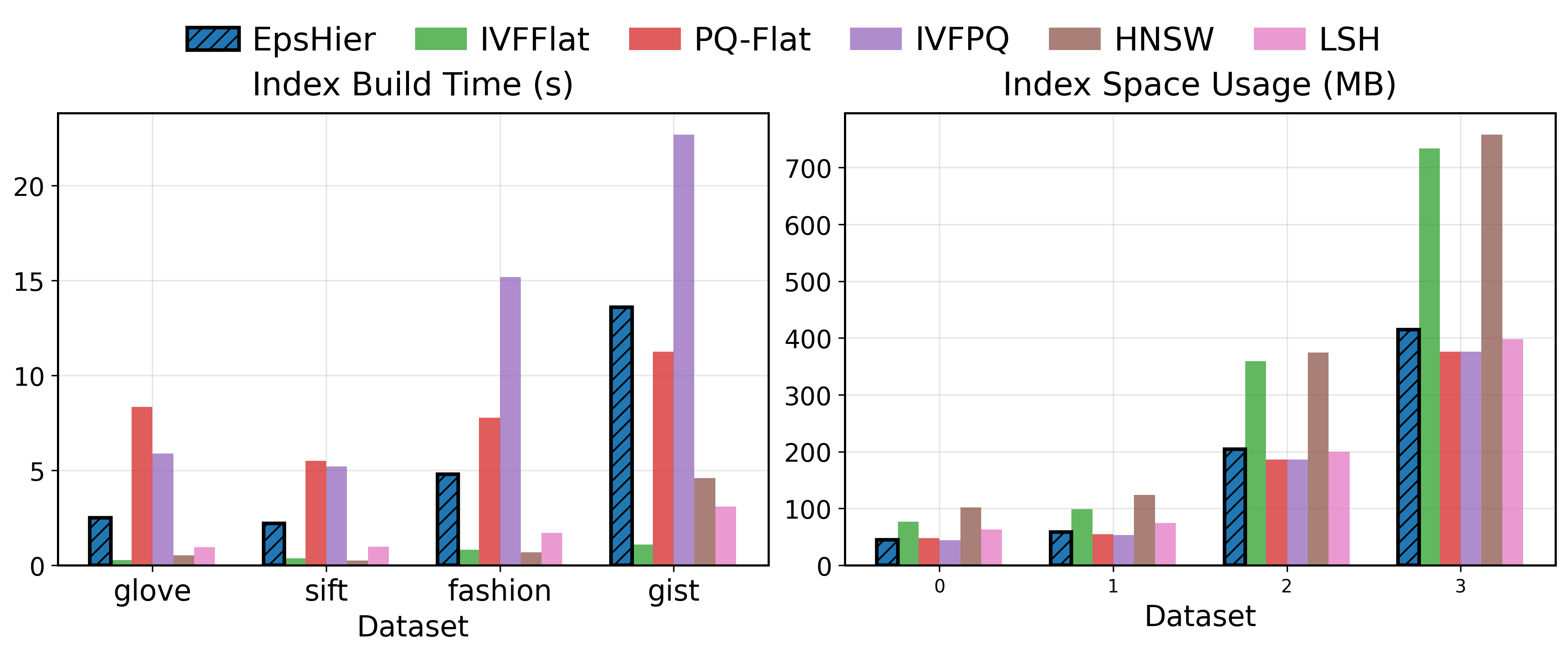}
    \caption{Comparing the space usage and building (indexing) time for Random-Access Similarity Search task}
    \label{fig:ann:preprocess}
\end{figure}

\subsection{Comparison with Exact Nearest Neighbor Search}
\label{app:faiss-flat}

We compare \textsc{EpsHier} against \texttt{FAISS IndexFlat}, a widely-used exact nearest-neighbor index that leverages SIMD-optimized distance computations. Since \texttt{IndexFlat} returns the top-$k$ nearest neighbors, retrieving the element at rank $i$ requires computing the top-$i$ results and selecting the last one.

Table~\ref{tab:faiss-flat} reports the query time on the GloVe-100-Angular dataset with one million vectors. Consistent with the trends observed in Figure~10, the runtime of exact nearest-neighbor retrieval increases as the target rank becomes larger, while the runtime of \textsc{EpsHier} remains nearly constant. Consequently, \textsc{EpsHier} becomes increasingly advantageous for accessing elements deeper in the ranked order, achieving up to $18.9\times$ speedup for large values of $i$.

\begin{table}[t]
\centering
\caption{Comparison with exact nearest-neighbor search using \texttt{FAISS IndexFlat} on the GloVe-100-Angular dataset ($n=10^6$).}
\label{tab:faiss-flat}
\begin{tabular}{lccc}
\toprule
Rank Fraction ($i/n$) & \textsc{EpsHier} (ms) & IndexFlat (ms) & Speedup \\
\midrule
$\approx 0.01\%$ & 8.8 & 1.4 & 0.2$\times$ \\
$\approx 0.1\%$  & 8.5 & 1.5 & 0.2$\times$ \\
$\approx 6\%$    & 9.1 & 20.8 & {\bf 2.3}$\times$ \\
$\approx 25\%$   & 8.8 & 74.4 & {\bf 8.4}$\times$ \\
$\approx 75\%$   & 9.1 & 172.1 & {\bf 18.9}$\times$ \\
\bottomrule
\end{tabular}
\end{table}

These results demonstrate that exact nearest-neighbor search remains highly effective for retrieving elements near the beginning of the ranked order. However, as the requested rank increases, the cost of retrieving the exact element grows substantially. In contrast, \textsc{EpsHier} maintains stable query performance across different rank positions, making it particularly suitable for random-access similarity search.

\subsection{Performance on Sparse Data}
\label{app:sparse}

To evaluate the behavior of our method on sparse high-dimensional data, we conducted an additional experiment on the TF-IDF representation of the 20 Newsgroups dataset, containing $n=10^4$ documents and $d=5000$ dimensions. We compare \textsc{EpsHier} against the exhaustive baseline used throughout the paper.

Table~\ref{tab:sparse} reports the query time for different target ranks. Similar to the results on dense datasets, \textsc{EpsHier} maintains nearly constant query time across different values of $i$ and consistently outperforms exhaustive search. 

\begin{table}[t]
\centering
\caption{Performance on the sparse TF-IDF 20 Newsgroups dataset ($n=10^4$, $d=5000$).}
\label{tab:sparse}
\begin{tabular}{ccc}
\toprule
Rank $i$ & \textsc{EpsHier} (ms) & Exhaustive (ms) \\
\midrule
156  & {\bf 15.4} & 43.2 \\
2500 & {\bf 13.8} & 43.6 \\
7500 & {\bf 13.5} & 44.0 \\
\bottomrule
\end{tabular}
\end{table}

\subsection{Effect of the Decay Rate $r$}

The decay rate $r$ controls the tradeoff between hierarchy depth and pruning effectiveness. Larger values of $r$ produce shallower hierarchies with fewer visited nodes, but also generate larger enclosing balls that overlap more frequently with query stripes, reducing pruning effectiveness. Smaller values of $r$ improve pruning quality at the cost of deeper traversals. Table~\ref{tab:r_effect} reports the impact of $r$ on the FIFA dataset. The results suggest that a moderate value ($r\approx4$) provides the best balance between traversal depth and pruning efficiency, yielding the lowest query time.

\begin{table}[t]
\centering
\caption{Effect of the decay rate $r$ on the FIFA dataset.}
\label{tab:r_effect}
\begin{tabular}{cccc}
\toprule
$r$ & \#Layers & Time (ms) & \#Visited Balls \\
\midrule
2   & 14 & 6.75 & 3053 \\
4   & 7  & {\bf 5.70} & 1796 \\
64  & 3  & 6.47 & 157 \\
128 & 2  & 6.64 & 78 \\
\bottomrule
\end{tabular}
\end{table}

\begin{table*}[t]
\centering
\footnotesize
\setlength{\tabcolsep}{3pt}
\begin{tabular}{lr@{\hspace{4pt}}r@{\hspace{8pt}}r@{\hspace{4pt}}r@{\hspace{8pt}}r@{\hspace{4pt}}r@{\hspace{8pt}}r@{\hspace{4pt}}r@{\hspace{8pt}}r@{\hspace{4pt}}r@{\hspace{8pt}}r@{\hspace{4pt}}r@{\hspace{8pt}}r@{\hspace{4pt}}r@{\hspace{8pt}}r@{\hspace{4pt}}r}
\toprule
Method & \textit{Time (s)} & \textit{Recall} & \textit{Time (s)} & \textit{Recall} & \textit{Time (s)} & \textit{Recall} & \textit{Time (s)} & \textit{Recall} & \textit{Time (s)} & \textit{Recall} & \textit{Time (s)} & \textit{Recall} & \textit{Time (s)} & \textit{Recall} & \textit{Time (s)} & \textit{Recall} \\
\cmidrule(lr){2-3}\cmidrule(lr){4-5}\cmidrule(lr){6-7}\cmidrule(lr){8-9}\cmidrule(lr){10-11}\cmidrule(lr){12-13}\cmidrule(lr){14-15}\cmidrule(lr){16-17}
\multicolumn{17}{l}{\textbf{fashion-mnist-784-euclidean} ($n=60,000$)} \\
 & \multicolumn{2}{c}{$i=10$} & \multicolumn{2}{c}{$i=29$ ($n/2^{11}$)} & \multicolumn{2}{c}{$i=117$ ($n/2^{9}$)} & \multicolumn{2}{c}{$i=468$ ($n/2^{7}$)} & \multicolumn{2}{c}{$i=1875$ ($n/2^{5}$)} & \multicolumn{2}{c}{$i=7500$ ($n/2^{3}$)} & \multicolumn{2}{c}{$i=30000$ ($n/2^{1}$)} & \multicolumn{2}{c}{$i=45000$ ($3n/4$)} \\
\cmidrule(lr){2-3}\cmidrule(lr){4-5}\cmidrule(lr){6-7}\cmidrule(lr){8-9}\cmidrule(lr){10-11}\cmidrule(lr){12-13}\cmidrule(lr){14-15}\cmidrule(lr){16-17}
\textbf{EpsHier* (Ours)} & 0.018 & \colorbox{green!20}{\textbf{1.000}} & 0.021 & \colorbox{green!20}{\textbf{1.000}} & 0.021 & \colorbox{green!20}{\textbf{1.000}} & 0.020 & \colorbox{green!20}{\textbf{1.000}} & 0.022 & \colorbox{green!20}{\textbf{1.000}} & \textbf{0.021} & \colorbox{green!20}{\textbf{1.000}} & \textbf{0.019} & \colorbox{green!20}{\textbf{0.920}} & \textbf{0.023} & \colorbox{green!20}{\textbf{1.000}} \\
Quick Select & 0.114 & \colorbox{green!20}{\textbf{1.000}} & 0.113 & \colorbox{green!20}{\textbf{1.000}} & 0.112 & \colorbox{green!20}{\textbf{1.000}} & 0.111 & \colorbox{green!20}{\textbf{1.000}} & 0.111 & \colorbox{green!20}{\textbf{1.000}} & 0.113 & \colorbox{green!20}{\textbf{1.000}} & 0.117 & \colorbox{green!20}{\textbf{1.000}} & 0.117 & \colorbox{green!20}{\textbf{1.000}} \\
IVFFlat & 0.003 & \colorbox{green!20}{\textbf{1.000}} & 0.003 & \colorbox{green!20}{\textbf{1.000}} & \textbf{0.003} & \colorbox{green!20}{\textbf{1.000}} & \textbf{0.004} & \colorbox{green!20}{\textbf{1.000}} & \textbf{0.007} & \colorbox{green!20}{\textbf{1.000}} & 0.018 & \colorbox{red!20}{0.000} & 0.065 & \colorbox{red!20}{0.000} & 0.069 & \colorbox{red!20}{0.000} \\
PQ & 0.003 & \colorbox{green!20}{\textbf{0.920}} & 0.003 & 0.790 & 0.004 & 0.670 & 0.005 & 0.470 & 0.010 & 0.530 & 0.028 & \colorbox{red!20}{0.000} & 0.080 & \colorbox{red!20}{0.000} & 0.111 & \colorbox{red!20}{0.000} \\
IVFPQ & 0.001 & 0.810 & 0.001 & 0.610 & 0.001 & 0.400 & 0.002 & 0.450 & 0.007 & 0.420 & 0.026 & \colorbox{red!20}{0.000} & 0.030 & \colorbox{red!20}{0.000} & 0.032 & \colorbox{red!20}{0.000} \\
HNSW & \textbf{0.000} & \colorbox{green!20}{\textbf{1.000}} & \textbf{0.000} & \colorbox{green!20}{\textbf{0.990}} & 0.001 & 0.880 & 0.002 & 0.330 & 0.002 & \colorbox{red!20}{0.050} & 0.002 & \colorbox{red!20}{0.000} & 0.005 & \colorbox{red!20}{0.000} & 0.006 & \colorbox{red!20}{0.000} \\
LSH & 0.053 & 0.390 & 0.060 & 0.270 & 0.060 & 0.250 & 0.059 & 0.210 & 0.056 & 0.270 & 0.041 & \colorbox{red!20}{0.030} & 0.085 & \colorbox{red!20}{0.000} & 0.116 & \colorbox{red!20}{0.000} \\
\midrule
\multicolumn{17}{l}{\textbf{glove-100-angular} ($n=1,183,514$)} \\
 & \multicolumn{2}{c}{$i=10$} & \multicolumn{2}{c}{$i=144$ ($n/2^{13}$)} & \multicolumn{2}{c}{$i=2311$ ($n/2^{9}$)} & \multicolumn{2}{c}{$i=9246$ ($n/2^{7}$)} & \multicolumn{2}{c}{$i=36984$ ($n/2^{5}$)} & \multicolumn{2}{c}{$i=147939$ ($n/2^{3}$)} & \multicolumn{2}{c}{$i=591757$ ($n/2^{1}$)} & \multicolumn{2}{c}{$i=887635$ ($3n/4$)} \\
\cmidrule(lr){2-3}\cmidrule(lr){4-5}\cmidrule(lr){6-7}\cmidrule(lr){8-9}\cmidrule(lr){10-11}\cmidrule(lr){12-13}\cmidrule(lr){14-15}\cmidrule(lr){16-17}
\textbf{EpsHier* (Ours)} & 0.022 & \colorbox{green!20}{\textbf{1.000}} & 0.026 & \colorbox{green!20}{\textbf{1.000}} & 0.027 & \colorbox{green!20}{\textbf{1.000}} & \textbf{0.023} & \colorbox{green!20}{\textbf{1.000}} & \textbf{0.021} & \colorbox{green!20}{\textbf{1.000}} & \textbf{0.022} & \colorbox{green!20}{\textbf{1.000}} & \textbf{0.022} & \colorbox{green!20}{\textbf{1.000}} & \textbf{0.022} & \colorbox{green!20}{\textbf{1.000}} \\
Quick Select & 3.494 & \colorbox{green!20}{\textbf{1.000}} & 3.386 & \colorbox{green!20}{\textbf{1.000}} & 3.393 & \colorbox{green!20}{\textbf{1.000}} & 3.389 & \colorbox{green!20}{\textbf{1.000}} & 3.422 & \colorbox{green!20}{\textbf{1.000}} & 3.502 & \colorbox{green!20}{\textbf{1.000}} & 3.606 & \colorbox{green!20}{\textbf{1.000}} & 3.616 & \colorbox{green!20}{\textbf{1.000}} \\
IVFFlat & 0.006 & \colorbox{green!20}{\textbf{0.990}} & \textbf{0.006} & \colorbox{green!20}{\textbf{0.980}} & \textbf{0.011} & \colorbox{green!20}{\textbf{0.960}} & 0.027 & \colorbox{green!20}{\textbf{0.950}} & 0.089 & 0.840 & 0.339 & \colorbox{red!20}{0.000} & 1.261 & \colorbox{red!20}{0.000} & 1.352 & \colorbox{red!20}{0.000} \\
PQ & 0.044 & \colorbox{green!20}{\textbf{1.000}} & 0.044 & \colorbox{green!20}{\textbf{0.920}} & 0.050 & 0.600 & 0.063 & 0.550 & 0.128 & 0.600 & 0.402 & \colorbox{red!20}{0.000} & 1.432 & \colorbox{red!20}{0.000} & 2.078 & \colorbox{red!20}{0.000} \\
IVFPQ & \textbf{0.004} & \colorbox{green!20}{\textbf{0.980}} & 0.005 & 0.530 & 0.012 & 0.590 & 0.036 & 0.510 & 0.104 & 0.560 & 0.339 & 0.210 & 0.387 & \colorbox{red!20}{0.000} & 0.404 & \colorbox{red!20}{0.000} \\
HNSW & 0.002 & 0.810 & 0.001 & 0.490 & 0.005 & \colorbox{red!20}{0.100} & 0.005 & \colorbox{red!20}{0.000} & 0.008 & \colorbox{red!20}{0.000} & 0.021 & \colorbox{red!20}{0.000} & 0.067 & \colorbox{red!20}{0.000} & 0.090 & \colorbox{red!20}{0.000} \\
LSH & 0.056 & 0.740 & 0.043 & 0.400 & 0.049 & 0.400 & 0.056 & 0.360 & 0.130 & 0.400 & 0.431 & \colorbox{red!20}{0.000} & 1.538 & \colorbox{red!20}{0.000} & 2.227 & \colorbox{red!20}{0.000} \\
\midrule
\multicolumn{17}{l}{\textbf{glove-25-angular} ($n=1,183,514$)} \\
 & \multicolumn{2}{c}{$i=10$} & \multicolumn{2}{c}{$i=144$ ($n/2^{13}$)} & \multicolumn{2}{c}{$i=2311$ ($n/2^{9}$)} & \multicolumn{2}{c}{$i=9246$ ($n/2^{7}$)} & \multicolumn{2}{c}{$i=36984$ ($n/2^{5}$)} & \multicolumn{2}{c}{$i=147939$ ($n/2^{3}$)} & \multicolumn{2}{c}{$i=591757$ ($n/2^{1}$)} & \multicolumn{2}{c}{$i=887635$ ($3n/4$)} \\
\cmidrule(lr){2-3}\cmidrule(lr){4-5}\cmidrule(lr){6-7}\cmidrule(lr){8-9}\cmidrule(lr){10-11}\cmidrule(lr){12-13}\cmidrule(lr){14-15}\cmidrule(lr){16-17}
\textbf{EpsHier* (Ours)} & 0.176 & \colorbox{green!20}{\textbf{1.000}} & 0.170 & \colorbox{green!20}{\textbf{1.000}} & 0.168 & \colorbox{green!20}{\textbf{1.000}} & 0.169 & \colorbox{green!20}{\textbf{1.000}} & 0.159 & \colorbox{green!20}{\textbf{1.000}} & \textbf{0.021} & \colorbox{green!20}{\textbf{1.000}} & \textbf{0.013} & \colorbox{green!20}{\textbf{1.000}} & \textbf{0.016} & \colorbox{green!20}{\textbf{1.000}} \\
Quick Select & 3.470 & \colorbox{green!20}{\textbf{1.000}} & 3.360 & \colorbox{green!20}{\textbf{1.000}} & 3.371 & \colorbox{green!20}{\textbf{1.000}} & 3.382 & \colorbox{green!20}{\textbf{1.000}} & 3.391 & \colorbox{green!20}{\textbf{1.000}} & 3.486 & \colorbox{green!20}{\textbf{1.000}} & 3.574 & \colorbox{green!20}{\textbf{1.000}} & 3.573 & \colorbox{green!20}{\textbf{1.000}} \\
IVFFlat & 0.003 & \colorbox{green!20}{\textbf{1.000}} & 0.003 & \colorbox{green!20}{\textbf{1.000}} & \textbf{0.007} & \colorbox{green!20}{\textbf{1.000}} & \textbf{0.021} & \colorbox{green!20}{\textbf{0.980}} & \textbf{0.080} & \colorbox{green!20}{\textbf{0.960}} & 0.327 & \colorbox{red!20}{0.000} & 1.227 & \colorbox{red!20}{0.000} & 1.299 & \colorbox{red!20}{0.000} \\
PQ & 0.013 & \colorbox{green!20}{\textbf{0.990}} & 0.014 & 0.710 & 0.021 & 0.470 & 0.040 & 0.570 & 0.097 & 0.450 & 0.373 & \colorbox{red!20}{0.000} & 1.393 & \colorbox{red!20}{0.000} & 2.032 & \colorbox{red!20}{0.000} \\
IVFPQ & 0.003 & \colorbox{green!20}{\textbf{1.000}} & \textbf{0.003} & \colorbox{green!20}{\textbf{0.940}} & 0.010 & 0.680 & 0.034 & 0.800 & 0.103 & 0.640 & 0.330 & \colorbox{red!20}{0.060} & 0.383 & \colorbox{red!20}{0.000} & 0.402 & \colorbox{red!20}{0.000} \\
HNSW & \textbf{0.002} & \colorbox{green!20}{\textbf{0.960}} & 0.001 & 0.710 & 0.003 & \colorbox{red!20}{0.160} & 0.004 & \colorbox{red!20}{0.030} & 0.007 & \colorbox{red!20}{0.000} & 0.021 & \colorbox{red!20}{0.000} & 0.068 & \colorbox{red!20}{0.000} & 0.093 & \colorbox{red!20}{0.000} \\
LSH & 0.150 & 0.880 & 0.147 & 0.520 & 0.138 & 0.440 & 0.131 & 0.540 & 0.196 & 0.470 & 0.472 & \colorbox{red!20}{0.000} & 1.544 & \colorbox{red!20}{0.000} & 2.213 & \colorbox{red!20}{0.000} \\
\midrule
\multicolumn{17}{l}{\textbf{sift-128-euclidean} ($n=1,000,000$)} \\
 & \multicolumn{2}{c}{$i=10$} & \multicolumn{2}{c}{$i=122$ ($n/2^{13}$)} & \multicolumn{2}{c}{$i=1953$ ($n/2^{9}$)} & \multicolumn{2}{c}{$i=7812$ ($n/2^{7}$)} & \multicolumn{2}{c}{$i=31250$ ($n/2^{5}$)} & \multicolumn{2}{c}{$i=125000$ ($n/2^{3}$)} & \multicolumn{2}{c}{$i=500000$ ($n/2^{1}$)} & \multicolumn{2}{c}{$i=750000$ ($3n/4$)} \\
\cmidrule(lr){2-3}\cmidrule(lr){4-5}\cmidrule(lr){6-7}\cmidrule(lr){8-9}\cmidrule(lr){10-11}\cmidrule(lr){12-13}\cmidrule(lr){14-15}\cmidrule(lr){16-17}
\textbf{EpsHier* (Ours)} & 0.220 & \colorbox{green!20}{\textbf{1.000}} & 0.222 & \colorbox{green!20}{\textbf{1.000}} & 0.228 & \colorbox{green!20}{\textbf{1.000}} & 0.233 & \colorbox{green!20}{\textbf{1.000}} & \textbf{0.236} & \colorbox{green!20}{\textbf{1.000}} & \textbf{0.220} & \colorbox{green!20}{\textbf{0.990}} & \textbf{0.026} & \colorbox{green!20}{\textbf{0.990}} & \textbf{0.204} & \colorbox{green!20}{\textbf{1.000}} \\
Quick Select & 1.707 & \colorbox{green!20}{\textbf{1.000}} & 1.633 & \colorbox{green!20}{\textbf{1.000}} & 1.630 & \colorbox{green!20}{\textbf{1.000}} & 1.638 & \colorbox{green!20}{\textbf{1.000}} & 1.681 & \colorbox{green!20}{\textbf{1.000}} & 1.757 & \colorbox{green!20}{\textbf{1.000}} & 1.815 & \colorbox{green!20}{\textbf{1.000}} & 1.835 & \colorbox{green!20}{\textbf{1.000}} \\
IVFFlat & 0.006 & \colorbox{green!20}{\textbf{1.000}} & \textbf{0.007} & \colorbox{green!20}{\textbf{1.000}} & \textbf{0.011} & \colorbox{green!20}{\textbf{1.000}} & \textbf{0.022} & \colorbox{green!20}{\textbf{1.000}} & 0.069 & 0.610 & 0.250 & \colorbox{red!20}{0.000} & 0.922 & \colorbox{red!20}{0.000} & 0.962 & \colorbox{red!20}{0.000} \\
PQ & 0.022 & \colorbox{green!20}{\textbf{1.000}} & 0.022 & 0.590 & 0.028 & 0.480 & 0.043 & 0.410 & 0.086 & 0.370 & 0.296 & \colorbox{red!20}{0.000} & 1.083 & \colorbox{red!20}{0.000} & 1.603 & \colorbox{red!20}{0.000} \\
IVFPQ & 0.003 & 0.840 & 0.003 & 0.410 & 0.008 & 0.430 & 0.027 & 0.400 & 0.085 & 0.530 & 0.259 & \colorbox{red!20}{0.000} & 0.287 & \colorbox{red!20}{0.000} & 0.302 & \colorbox{red!20}{0.000} \\
HNSW & \textbf{0.001} & \colorbox{green!20}{\textbf{0.990}} & 0.001 & 0.760 & 0.003 & \colorbox{red!20}{0.110} & 0.004 & \colorbox{red!20}{0.040} & 0.006 & \colorbox{red!20}{0.000} & 0.017 & \colorbox{red!20}{0.000} & 0.056 & \colorbox{red!20}{0.000} & 0.081 & \colorbox{red!20}{0.000} \\
LSH & 0.136 & 0.700 & 0.134 & 0.430 & 0.123 & 0.430 & 0.121 & 0.430 & 0.161 & 0.520 & 0.370 & \colorbox{red!20}{0.000} & 1.189 & \colorbox{red!20}{0.000} & 1.709 & \colorbox{red!20}{0.000} \\
\midrule
\multicolumn{17}{l}{\textbf{gist-960-euclidean} ($n=1,000,000$)} \\
 & \multicolumn{2}{c}{$i=10$} & \multicolumn{2}{c}{$i=122$ ($n/2^{13}$)} & \multicolumn{2}{c}{$i=1953$ ($n/2^{9}$)} & \multicolumn{2}{c}{$i=7812$ ($n/2^{7}$)} & \multicolumn{2}{c}{$i=31250$ ($n/2^{5}$)} & \multicolumn{2}{c}{$i=125000$ ($n/2^{3}$)} & \multicolumn{2}{c}{$i=500000$ ($n/2^{1}$)} & \multicolumn{2}{c}{$i=750000$ ($3n/4$)} \\
\cmidrule(lr){2-3}\cmidrule(lr){4-5}\cmidrule(lr){6-7}\cmidrule(lr){8-9}\cmidrule(lr){10-11}\cmidrule(lr){12-13}\cmidrule(lr){14-15}\cmidrule(lr){16-17}
\textbf{EpsHier* (Ours)} & 0.497 & \colorbox{green!20}{\textbf{1.000}} & 0.492 & \colorbox{green!20}{\textbf{1.000}} & 0.496 & \colorbox{green!20}{\textbf{1.000}} & 0.497 & \colorbox{green!20}{\textbf{1.000}} & 0.458 & \colorbox{green!20}{\textbf{1.000}} & \textbf{0.388} & \colorbox{green!20}{\textbf{1.000}} & \textbf{0.256} & \colorbox{green!20}{\textbf{1.000}} & \textbf{0.363} & \colorbox{green!20}{\textbf{1.000}} \\
Quick Select & 2.131 & \colorbox{green!20}{\textbf{1.000}} & 2.019 & \colorbox{green!20}{\textbf{1.000}} & 2.015 & \colorbox{green!20}{\textbf{1.000}} & 2.022 & \colorbox{green!20}{\textbf{1.000}} & 2.049 & \colorbox{green!20}{\textbf{1.000}} & 2.126 & \colorbox{green!20}{\textbf{1.000}} & 2.238 & \colorbox{green!20}{\textbf{1.000}} & 2.266 & \colorbox{green!20}{\textbf{1.000}} \\
IVFFlat & \textbf{0.056} & \colorbox{green!20}{\textbf{1.000}} & \textbf{0.057} & \colorbox{green!20}{\textbf{1.000}} & \textbf{0.062} & \colorbox{green!20}{\textbf{1.000}} & \textbf{0.078} & \colorbox{green!20}{\textbf{0.990}} & \textbf{0.144} & \colorbox{green!20}{\textbf{0.990}} & 0.386 & \colorbox{red!20}{0.000} & 1.256 & \colorbox{red!20}{0.000} & 1.510 & \colorbox{red!20}{0.000} \\
PQ & 0.035 & 0.430 & 0.035 & 0.290 & 0.042 & 0.280 & 0.056 & 0.300 & 0.118 & 0.440 & 0.375 & \colorbox{red!20}{0.000} & 1.357 & \colorbox{red!20}{0.000} & 2.003 & \colorbox{red!20}{0.000} \\
IVFPQ & 0.005 & 0.420 & 0.005 & 0.200 & 0.012 & 0.260 & 0.033 & 0.340 & 0.097 & 0.380 & 0.320 & \colorbox{red!20}{0.060} & 0.451 & \colorbox{red!20}{0.000} & 0.466 & \colorbox{red!20}{0.000} \\
HNSW & 0.002 & 0.770 & 0.001 & 0.490 & 0.006 & \colorbox{red!20}{0.090} & 0.006 & \colorbox{red!20}{0.030} & 0.009 & \colorbox{red!20}{0.000} & 0.019 & \colorbox{red!20}{0.000} & 0.059 & \colorbox{red!20}{0.000} & 0.079 & \colorbox{red!20}{0.000} \\
LSH & 0.125 & \colorbox{red!20}{0.160} & 0.116 & \colorbox{red!20}{0.150} & 0.109 & \colorbox{red!20}{0.110} & 0.127 & \colorbox{red!20}{0.150} & 0.177 & \colorbox{red!20}{0.120} & 0.429 & \colorbox{red!20}{0.100} & 1.462 & \colorbox{red!20}{0.000} & 2.123 & \colorbox{red!20}{0.000} \\
\bottomrule
\end{tabular}
\caption{Query time and recall at different $i$ values for various datasets. All values shown with 3 decimal precision. Recall values greater than 0.9 are shown in bold with \colorbox{green!20}{green highlight}, and values less than 0.2 are shown with \colorbox{red!20}{red highlight}. The fastest query time among methods with recall > 0.9 is shown in bold for each $i$.}
\label{tab:knn-time-recall-comparison}
\end{table*}

\end{document}